\pdfoutput=1
\RequirePackage{ifpdf}
\ifpdf 
\documentclass[pdftex]{sigma}
\else
\documentclass{sigma}
\fi

\numberwithin{equation}{section}

\newtheorem{Theorem}{Theorem}[section]
\newtheorem*{Theorem*}{Theorem}
\newtheorem{Corollary}[Theorem]{Corollary}
\newtheorem{Lemma}[Theorem]{Lemma}
\newtheorem{Proposition}[Theorem]{Proposition}
 { \theoremstyle{definition}
\newtheorem{Definition}[Theorem]{Definition}

\newtheorem{Example}[Theorem]{Example}
\newtheorem{Remark}[Theorem]{Remark} }

\usepackage{bbm,enumerate}

\newcommand{\im} {\mbox {Im}\hskip 0.5truemm}

\newcommand{\tint}{{\textstyle\int}}
\newcommand{\id}{\mathbbm{1}}

\newcommand{\ba}{\begin{align}}
\newcommand{\ea}{\end{align}}
\newcommand{\nn}{\nonumber}

\newcommand{\R}{{\mathcal R}}

\newcommand{\tr}{{\rm tr}}

\newcommand{\bt}{{\bf t}}

\newcommand{\p}{\partial}

\newcommand{\g}{{\mathfrak g}}

\newcommand{\mf}[1]{\mathfrak{#1}}
\newcommand{\mb}[1]{\mathbb{#1}}
\newcommand{\mc}[1]{\mathcal{#1}}

\DeclareMathOperator{\ad}{ad}
\DeclareMathOperator{\Ker}{Ker}

\DeclareMathOperator{\Res}{Res}

\def\wbigoplus{\mathop{\widehat{\bigoplus}}\limits}

\def\={\; = \;}
\def\+{\, + \,}
\def\:={\; := \; }

\usepackage{mathtools}
\usepackage{tikz}
\usetikzlibrary{chains}

\tikzset{node distance=2em, ch/.style={circle,draw,on chain,inner sep=2pt},chj/.style={ch,join},every path/.style={shorten >=4pt,shorten <=4pt},line width=1pt,baseline=-1ex}

\let\dlabel=\alabel
\let\ulabel=\mlabel

\newcommand{\dnode}[2][chj]{%
\node[#1,label={below:\dlabel{#2}}] {};
}

\newcommand{\dnodea}[3][chj]{%
\dnode[#1,label={above:\ulabel{#2}}]{#3}
}

\newcommand{\dnodeanj}[2]{%
\dnodea[ch]{#1}{#2}
}

\newcommand{\QLeftarrow}{%
\begingroup
\tikz
\draw[shorten >=0pt,shorten <=0pt] (0,3pt) -- ++(-1em,0) (0,1pt) -- ++(-1em-1pt,0) (0,-1pt) -- ++(-1em-1pt,0) (0,-3pt) -- ++(-1em,0) (-1em+1pt,5pt) to[out=-105,in=45] (-1em-2pt,0) to[out=-45,in=105] (-1em+1pt,-5pt);
\endgroup
}

\begin{document}
\allowdisplaybreaks

\newcommand{\arXivNumber}{2110.06655}

\renewcommand{\PaperNumber}{077}

\FirstPageHeading

\ShortArticleName{Affine Kac--Moody Algebras and Tau-Functions for the Drinfeld--Sokolov Hierarchies}

\ArticleName{Affine Kac--Moody Algebras and Tau-Functions\\ for the Drinfeld--Sokolov Hierarchies:\\ the Matrix-Resolvent Method}

\Author{Boris DUBROVIN~$^{\rm a}$, Daniele VALERI~$^{\rm bc}$ and Di YANG~$^{\rm d}$}

\AuthorNameForHeading{B.~Dubrovin, D.~Valeri and D.~Yang}

\Address{$^{\rm a)}$~Deceased}

\Address{$^{\rm b)}$~Dipartimento di Matematica, Sapienza Universit\`a di Roma,\\
\hphantom{$^{\rm b)}$}~P.le Aldo Moro 5, 00185 Rome, Italy}
\EmailD{\href{mailto:daniele.valeri@uniroma1.it}{daniele.valeri@uniroma1.it}}
\URLaddressD{\url{https://danielevaleri.site.uniroma1.it/}}
\Address{$^{\rm c)}$~INFN, Section of Rome, Italy}

\Address{$^{\rm d)}$~School of Mathematical Sciences, USTC, Hefei 230026, P.R.~China}
\EmailD{\href{mailto:diyang@ustc.edu.cn}{diyang@ustc.edu.cn}}

\ArticleDates{Received April 07, 2022, in final form September 26, 2022; Published online October 14, 2022}

\Abstract{For each affine Kac--Moody algebra $X_n^{(r)}$ of rank $\ell$, $r=1,2$, or $3$, and for every choice of a vertex~$c_m$, $m=0,\dots,\ell$, of the corresponding Dynkin diagram, by using the matrix-resolvent method we define a gauge-invariant tau-structure for the associated Drinfeld--Sokolov hierarchy and give explicit formulas for generating series of logarithmic derivatives of the tau-function in terms of matrix resolvents, extending the results of [\textit{Mosc. Math.~J.} \textbf{21} (2021), 233--270, arXiv:1610.07534] with $r=1$ and $m=0$. For the case $r=1$ and $m=0$, we verify that the above-defined tau-structure agrees with the axioms of Hamiltonian tau-symmetry in the sense of [\textit{Adv. Math.} \textbf{293} (2016), 382--435, arXiv:1409.4616] and [arXiv:math.DG/0108160].}


\Classification{37K10; 17B80; 17B67; 37K30}

\section{Introduction}\label{section1}

Let $X_n^{(r)}$ be an affine Kac--Moody algebra of rank~$\ell$, with $r=1,2,3$ (here $n=n(\ell)$, for example, $n(\ell)=\ell$ when $r=1$),
and let
$C=(C_{ij})_{i,j=0}^\ell$ be its Cartan matrix~\cite{Kac94}.
In $X_{n}^{(r)}$ there is a set of
Chevalley generators $\{e_i,h_i,f_i\mid i=0,\dots,\ell\}$
satisfying the following relations:
\begin{equation}\label{rel:Lie}
[h_i,h_j]=0,\quad [e_i,f_j]=\delta_{ij}h_i,\quad[h_i,e_j]=C_{ij}e_j, \quad [h_i,f_j]=-C_{ij}f_j, \quad \forall\, 0\leq i,j \leq \ell,
\end{equation}
and for $i\neq j$ we have
\begin{equation}\label{rel:Serre}
(\ad e_i)^{1-C_{ij}}e_j=(\ad f_i)^{1-C_{ij}}f_j=0.
\end{equation}

Let $a_i$ (respectively $a_i^\vee$) be the positive integers satisfying
$\sum_{j=0}^\ell C_{ij}a_j=0$
\big(respectively $\sum_{j=0}^\ell C_{ji}a_j^\vee=0$\big), for all $i=0,\dots,\ell$, such that their greatest common divisor is~1.
The number{\samepage
\begin{equation*}
h=\sum_{i=0}^\ell a_i
\qquad
\bigg(\text{respectively }
h^\vee=\sum_{i=0}^\ell a_i^\vee\bigg)
\end{equation*}
is called \emph{Coxeter number} (respectively \emph{dual Coxeter number}) of $X_{n}^{(r)}$~\cite{Kac94}.}

Let $\widetilde{\mf g}$ be the quotient of $X_{n}^{(r)}$ by the one-dimensional space generated by the
central element
$K=\sum_{i=0}^\ell a_i^\vee h_i$.
The \emph{principal gradation} on $\widetilde{\mf g}$ is defined by
assigning $\deg^{\rm pr} e_i=-\deg^{\rm pr} f_i=1$, $i=0,\dots,\ell$. Clearly, $\widetilde{\mf g}$ decomposes into the direct sum of homogeneous subspaces
\[
\widetilde{\mf g}=\bigoplus_{k\in\mb Z}\widetilde{\mf g}^{k}
,
\]
where elements in $\widetilde{\mf g}^{\,k}$ have principal degree $k$.
In this paper we are interested with a completion of $\widetilde{\mf g}$ rather than with $\widetilde{\mf g}$ itself.
By an abuse of notation we denote these two objects with the same symbol and let
\begin{equation}\label{deg:principal}
\widetilde{\mf g}=\wbigoplus_{k\in\mb Z}\,\widetilde{\mf g}^{k}
,
\end{equation}
where the direct sum is completed by allowing infinite series in negative degree.
Given an element $a\in\widetilde{\mf g}$ we denote by $a^+$ its projection on $\widetilde{\mf g}^{\geq0}=\oplus_{k\geq0}\widetilde{\mf g}^k$ and
by $a^{-}$ its projection on $\widetilde{\mf g}^{<0}=\widehat{\oplus}_{k<0}\widetilde{\mf g}^k$.

Introduce the {\it cyclic element}
\begin{equation}\label{eq:Lambda}
\Lambda=\sum_{i=0}^\ell e_i\in\widetilde{\mf g}.
\end{equation}
Note that $\Lambda$ is homogeneous of principal degree~$1$.
Let $\mc H=\Ker\ad\Lambda$ be the so-called \emph{principal $($centerless$)$ Heisenberg subalgebra} of $\widetilde{\mf g}$.
According to \cite{Kac78} (cf.~\cite{Kos}), $\mc H$ is abelian
and we have the direct sum decomposition
\begin{equation}\label{dec:Lambda}
\widetilde{\mf g}=\mc H\oplus\im\ad\Lambda
.
\end{equation}
Given $A\in \widetilde{\mf g}$ we denote by $\pi_{\mc H}(A)\in\mc H$ its projection with respect to the direct sum decomposition~\eqref{dec:Lambda} (namely $\Ker\pi_{\mc H}=\im\ad\Lambda$).

It is known that
$\mc H$ and $\im\ad\Lambda$ admit the following decomposition:
\begin{equation}
\label{mcHdecompbasislambda}
\mc H=\wbigoplus_{i\in E}\mb C \Lambda_i,\qquad
\mf\im\ad\Lambda=\wbigoplus_{i\in\mb Z}\,(\im\ad\Lambda)^i,
\end{equation}
where
$E\subset\mb Z$ is the set of exponents of $X_{n}^{(r)}$ (see~\cite{Kac94} for the definition of exponents),
$\deg^{\rm pr} \Lambda_i=i$, and
$(\im\ad\Lambda)^i=\im\ad\Lambda\cap\widetilde{\mf g}^{\,i}$, $i\in \mb Z$.
A convenient normalization of the basis elements~$\Lambda_i$ can be found in Section~\ref{sec:3.2}.
Recall that the set~$E$ has the following form
\begin{equation}\label{exponents_set}
E=\bigsqcup_{a=1}^n(m_a+rh\mb Z)
,
\end{equation}
where
\[
1=m_1<m_2\leq\cdots\leq m_{n-1}<m_n=rh-1
,
\]
satisfy the following relation:
\begin{equation}\label{20171106:eq2}
m_{a}+m_{n+1-a}=rh
,
\qquad
a=1,\dots,n
.
\end{equation}

Let $m$ be an integer from $0$ to $\ell$.
Take $c_m$ to be the $m$th vertex of the Dynkin diagram of~$X_{n}^{(r)}$.
Here we label the Dynkin diagram according to~\cite{Kac94}. Recall that
the vertex $c_0$ is the so-called {\it special vertex}~\cite{Kac94}.
The \emph{standard gradation} corresponding to
$c_m$ is defined by assigning $\deg_{\rm st} e_m=-\deg_{\rm st} f_m=1$, and degree~$0$ to all the remaining Chevalley generators.
Then, we also have the direct sum decomposition
\begin{equation}\label{deg:standard}
\widetilde{\mf g}=\wbigoplus_{k\in\mb Z}\,\widetilde{\mf g}_{k}
,
\end{equation}
where $\widetilde{\mf g}_k$ denotes the homogeneous subspace of elements with standard degree $k$.
Given an element $a\in\widetilde{\mf g}$ we denote by $a_+$ its projection on $\widetilde{\mf g}_{\geq0}=\oplus_{k\geq0}\widetilde{\mf g}_k$ and
by $a_{-}$ its projection on $\widetilde{\mf g}_{<0}=\widehat{\oplus}_{k<0}\widetilde{\mf g}_k$.

Denote $\mf a=\widetilde{\mf g}_0$. Note that $\mf a$ is a semisimple Lie algebra with Cartan matrix
$C_{\mf a}=(C_{ij})_{i,j\neq m}$ and Chevalley generators
$e_i$, $h_i$ and $f_i$, $i\in\{0,\dots,\ell\}\setminus\{m\}$. With respect to the principal gradation we can write
\begin{equation}\label{dec:dynkin}
\mf a=\bigoplus_{i=-h_{\mf a}+1}^{h_{\mf a}-1}\mf a^i
,
\end{equation}
where $\mf a^i=\mf a\cap\widetilde{\mf g}^{\,i}$ and $h_{\mf a}$ is the Coxeter number of $\mf a$.
In particular, $\mf a^0$, that is generated by $h_i$, $i\in\{0,\dots,\ell\}\setminus\{m\}$, is equal to $\widetilde{\mf g}^{\,0}$,
and it is a Cartan subalgebra of~$\mf a$.
Let us further denote by~$\mf n$ the nilpotent subalgebra $\mf n=\mf a^{<0}$ of~$\mf a$ and by~$\mf b$ the
Borel subalgebra $\mf b=\mf n\oplus\mf a^0$ of~$\mf a$. Clearly, $\mf n$~is generated by $f_i$,
$i\in\{0,\dots,\ell\}\setminus\{m\}$, and $\mf b$ is generated by $f_i$, $h_i$,
$i\in\{0,\dots,\ell\}\setminus\{m\}$.
From the defining relations~\eqref{rel:Lie} of the Kac--Moody algebra $\widetilde{\mf g}$ we have
\begin{equation}\label{20171109:eq1}
[\mf n,\mf b]\subset\mf n,\qquad
[\mf n,e_m]=0,\qquad
[\mf n,e_i]\subset\mf b,\qquad
i\in\{0,\dots,\ell\}\setminus\{m\}.
\end{equation}
Let $e=\Lambda-e_m\in\mf a$, where $\Lambda\in\widetilde{\mf g}$ is the cyclic element in~\eqref{eq:Lambda}.
The element~$e$ is called a {\it principal nilpotent element}.

Define a linear operator associated to the pair $(\widetilde{\mf g},c_m)$,
called a Lax operator, by
\begin{equation*}
\mc L=\partial+\Lambda+q,
\end{equation*}
where $\partial:=\partial_x$ and $q:=q(x)\in C^{\infty}\big(S^1,\mf b\big)$ is a smooth function from the circle to the Borel subalgebra~$\mf b$.
We denote by~$\mc A^q$
the algebra of differential polynomials in $q$, namely,
an element of~$\mc A^q$ is a polynomial in the entries of
the smooth function~$q$ and their $x$-derivatives.
Recall from \cite{DS85} that there exist a unique pair of elements $U\in\mc A^q\otimes(\im\ad\Lambda)^{<0}$ and
$H\in\mc A^q\otimes\mc H^{<0}$ such that
\begin{equation}\label{eq:L0}
{\rm e}^{\ad U}(\partial+\Lambda+q)=\partial+\Lambda+H.
\end{equation}
(Observe that the element $U$ used in this paper is $-U$ in \cite{BDY16}.)
\begin{Definition}\label{defresolvents}
An element $R\in\mc A^q\otimes\widetilde{\mf g}$ such that $[\mc L,R]=0$ is called a {\it resolvent}
for $\mc L$.
\end{Definition}
For every $i\in E$, we denote
\begin{equation}\label{eq:R_a}
R_i ={\rm e}^{-\ad U}(\Lambda_i)\in\mc A^q\otimes \widetilde{\mf g}.
\end{equation}
Clearly, $R_i$, $i\in E$, are resolvents for $\mc L$. Indeed, using \eqref{eq:L0} and the fact that $\mc H$ is abelian, we have
\begin{equation}\label{LRcomm}
[\mc L,R_i]
 =
{\rm e}^{-\ad U}[\partial+\Lambda+H,\Lambda_i]
= 0
.
\end{equation}
We call $R_{m_a}$, $a=1,\dots,n$, the {\it basic resolvents}.

Recall that the following system of evolutionary partial differential equations (PDEs)
\begin{equation}\label{eq:preDS}
\frac{\partial \mc L}{\partial t_i}= \bigl[(R_i)_+,\mc L\bigr],\qquad
i\in E\cap \mathbb{Z}_{>0},
\end{equation}
is called the \emph{pre-Drinfeld--Sokolov $($pre-DS$)$ hierarchy} associated to the pair $(\widetilde{\mf g},c_m)$.
The proof of the fact that~\eqref{eq:preDS} indeed defines evolutionary PDEs can be found in~\cite{DS85}.
Also according to~\cite{DS85}, the flows
in this system of PDEs all commute. Recall also that, for all $j\in E$ and $i\in E\cap \mathbb{Z}_{>0}$, we have
\begin{equation}\label{20171104:eq2}
\frac{\partial R_j}{\partial t_i}= [(R_i)_+, R_j],
\end{equation}

A \emph{gauge transformation}
is a change of variables $q\mapsto \tilde q \in\mc A^q\otimes\mf b$ of the form
\begin{equation}\label{gauge}
\widetilde{\mc L}={\rm e}^{\ad N}\mc L=\partial+\Lambda+\tilde q, \qquad N\in\mc A^q\otimes\mf n.
\end{equation}
Explicitly, we have
\begin{equation}\label{gauge2}
\tilde q=q-\sum_{k\geq1}\frac{(\ad N)^{k-1}}{k!}(\partial N)+\sum_{k\geq1}\frac{(\ad N)^k}{k!}(q+\Lambda).
\end{equation}
Due to the commutation relations in \eqref{20171109:eq1}, the expression for~$\tilde q$ in~\eqref{gauge2} is a~well-defined element of~$\mc A^q\otimes\mf b$.
Via the gauge transformation~\eqref{gauge2}, a resolvent~$R$ transforms as follows:
\begin{equation}\label{202110906rr}
R\mapsto {\rm e}^{\ad N} R =\widetilde{R}
.
\end{equation}

By a {\it gauge invariant}, we mean a differential polynomial $g(q,q_x,q_{xx},\dots)$ in~$\mc A^q$, such that $g(\tilde q, \tilde q_x,\tilde q_{xx},\dots)=g(q,q_x,q_{xx},\dots)$
for all gauge transformations~\eqref{gauge}.
The space of all gauge invariants, denoted by~$\mc R$, is a differential algebra~\cite{DS85} that can be identified with
the classical
$\mc W$-algebra~$\mc W(\mf a,e)$ associated to the Lie algebra $\mf a$ and its principal nilpotent element $e$~\cite{DSKV13}.

Since, $\ad e\colon \mf n\to\mf b$ is injective~\cite{DS85} (where we recall that $e$ is the principal nilpotent element),
we may choose a space $V\subset\mf b$ complementary to
$[e,\mf n]$ and compatible with the direct sum decomposition \eqref{dec:dynkin}, i.e.,
\begin{equation}\label{bVen}
\mf b= V \oplus [e,\mf n] .
\end{equation}
Note that
$\dim V=\dim\mf{b}-\dim\mf n=\dim\mf a^0=\ell$. The vector space $V\subset\mf b$ is called a \emph{Drinfeld--Sokolov $($DS$)$ gauge}.
It is proved by Drinfeld and Sokolov that there exists a unique $N^{\rm can}\in\mc A^q\otimes\mf n$ such that
\begin{equation}\label{Lcan}
\mc L^{\rm can}={\rm e}^{\ad N^{\rm can}}\mc L=\partial+\Lambda+q^{\rm can},
\qquad
q^{\rm can}\in\mc A^q\otimes V
.
\end{equation}
If $g(q,q_x,q_{xx},\dots)$ is an element in~$\mc R$, then $g(q,q_x,q_{xx},\dots)=g(q^{\rm can},q_x^{\rm can},q_{xx}^{\rm can},\dots)$. Hence $\mc R$ can be realized as
an algebra of polynomials in the entries $u_1,\dots,u_{\ell}$ of $q^{\rm can}$ and their $x$-derivatives.

By the results of~\cite{DS85}, the differential algebra~$\mc R$ is preserved by the flows of the pre-DS hierarchy
\eqref{eq:preDS}, namely, for every $i\in E\cap \mathbb{Z}_{>0}$ we have that
$\frac{\p}{\p t_i}(\mc R)\subset\mc R$.
\begin{Definition}
The \emph{DS hierarchy} associated to the affine Kac--Moody algebra $\widetilde{\mf g}$
and a vertex~$c_m$ of its Dynkin diagram is the set of equations
\begin{equation}\label{DS:hier}
\frac{\p u_s}{\p t_i} = P_{s,i}\in\mc R,
\qquad s=1,\dots,\ell,\quad i\in E\cap \mathbb{Z}_{>0},
\end{equation}
where the RHS of \eqref{DS:hier} can be computed by applying the flow in \eqref{eq:preDS}
to the gauge invariants~$u_s$.
\end{Definition}
It is known from~\cite{DS85} that
\[
\frac{\p u_s}{\p t_1} = -\frac{\p u_s}{\p x}.
\]
Therefore, for the DS hierarchy, we identify $t_1$ with~$-x$, and a solution $u_s(x,{\bf t})$
to the DS hierarchy~\eqref{DS:hier} will be simply denoted by $u_s({\bf t})$.
We also remark that, as it is shown in~\cite{DS85},
if~a~vertex of the Dynkin diagram of~$\widetilde{\mf g}$ is the image of the vertex~$c_m$
under an automorphism of the diagram, then the corresponding DS hierarchies
coincide.

The main theme of this paper is on computing logarithmic derivatives of tau-functions for the DS hierarchy using the
matrix-resolvent method~\cite{BDY16,BDY21,Zhou}.
 To proceed let us realize~$\widetilde{\g}$
as a~subalgebra $L(\mf g,\sigma_m)$ of~$L(\g)=\g\otimes\mb C\big(\big(\lambda^{-1}\big)\big)=\mf g\big(\big(\lambda^{-1}\big)\big)$, where $\g$ is a certain simple Lie algebra and $\sigma_m\colon \mf g\to\mf g$ is a finite-order automorphism
both depending on the pair $(\widetilde{\mf g},c_m)$ (see Section~\ref{sec:realizationKM} for the details). Denote $N_m=r a_m$, and let
$\pi_{\lambda}\colon \mb C\big(\big(\lambda^{-1}\big)\big)\to \lambda^{-1}\mb C\big[\big[\lambda^{-N_m}\big]\big]$ be defined by~\eqref{eq:pi} and $\pi_{\lambda,\mu}:=\pi_\mu\circ \pi_\lambda$.

\begin{Definition}
Define the series $F_{a,b}(\lambda,\mu)\in\mc A^{q}\otimes\lambda^{-1}\mu^{-1}\mb C\big[\big[\lambda^{-N_m},\mu^{-N_m}\big]\big]$,
$a,b=1,\dots,n$, by
\begin{equation}\label{eq:two-point-series}
F_{a,b}(\lambda,\mu)
= \pi_{\lambda,\mu}\bigg(
\frac{\big(R_{m_a}(\lambda)\big|R_{m_b}(\mu)\big)-\frac{\delta_{a+b,n+1}}{r} \big(m_a\lambda^{N_m}+m_b\mu^{N_m}\big)}{(\lambda-\mu)^2}\bigg).
\end{equation}
Here $(\cdot|\cdot)$ denotes the normalized Cartan--Killing form of~$\g$ with the natural extension to~$L(\g)$
(cf.~\eqref{ckdef1}--\eqref{ckdef2} for its precise definition).
\end{Definition}
The fact that the right-hand side of~\eqref{eq:two-point-series} belongs to
$\mc A^{q}\otimes\lambda^{-1}\mu^{-1}\mb C\big[\big[\lambda^{-N_m},\mu^{-N_m}\big]\big]$
will be proved in the beginning of Section~\ref{section33}.

Write
\begin{equation}\label{eq:Fab}
F_{a,b}(\lambda,\mu)=:\sum_{l,k\in\mb Z_{\geq0}} \Omega_{a,l;b,k}\lambda^{-N_ml-1}\mu^{-N_mk-1}.
\end{equation}
It follows from~\eqref{202110906rr} and the invariance property of~$(\cdot|\cdot)$ that the differential polynomials
$\Omega_{a,l;b,k}$, $a,b=1,\dots,n$, $l,k\in\mb Z_{\geq0}$,
defined via~\eqref{eq:Fab} belong to~$\mc R$. In particular,
$F_{a,b}(\lambda,\mu)$ does not change if we replace~$R_{m_c}$ in the right-hand side of~\eqref{eq:two-point-series} with
\begin{equation}\label{RcanNR}
R^{\rm can}_{m_c}:={\rm e}^{\ad N^{\rm can}} R_{m_c}
\end{equation}
(cf.~\eqref{202110906rr}).
We will prove in Section~\ref{section33} that
the differential polynomials~$\Omega_{a,l;b,k}\in\mc R$ also have the following properties:
\begin{align}
& \Omega_{a,l;b,k} = \Omega_{b,k;a,l}, \label{tau09081}\\
& \p_{t_{m_c+mrh}}\Omega_{a,l;b,k} = \p_{t_{m_a+lrh}} \Omega_{b,k;c,m} . \label{tau09082}
\end{align}
We call $\{\Omega_{a,l;b,k}\mid a,b=1,\dots,n,l,\,k\in\mb Z_{\geq0}\}$ the {\it tau-structure} for the DS hierarchy.
For $N\ge3$, $c_1,\dots,c_N\in \{1,\dots,n\}$, and $k_1,\dots,k_N\geq0$, we also define
\begin{equation}\label{defOmegaN}
\Omega_{c_1,k_1;\dots;c_N,k_N}:=
\p_{t_{m_{c_N}+k_Nrh}} \cdots \p_{t_{m_{c_3}+k_3rh}} (\Omega_{c_1,k_1;c_2,k_2}).
\end{equation}
Clearly, these elements all belong to~$\R$.
It follows from~\eqref{tau09081}--\eqref{tau09082} that $\Omega_{c_1,k_1;\dots;c_N,k_N}$ are totally symmetric
with respect to permuting its index-pairs.

For every $N\geq2$, we define a cyclic symmetric invariant $N$-linear form $B\colon \mf g\times\dots\times\mf g\to\mb C$ by
\[
B(x_1,\dots,x_{N})=\tr(\ad x_1\circ\dots\circ\ad x_{N}),
\qquad
x_1,\dots, x_{N}\in\mf g.
\]
We extend $B$ to a cyclic symmetric invariant linear $N$-form on $L(\mf g,\sigma_m)\times\dots\times L(\mf g,\sigma_m)$ with values in
$\mb C\big(\big(\lambda_1^{-1},\dots,\lambda_N^{-1}\big)\big)$ in the obvious way.
The main result of the paper is given by the following theorem.
\begin{Theorem}\label{thm:N-point}
For each $N\geq2$, let $c_1,\dots,c_N$ be arbitrarily given integers in~$\{1,\dots,n\}$.
We have
\begin{gather}
\sum_{k_1,\dots,k_N\in\mb Z_{\geq0}} \prod_{j=1}^N
\lambda_j^{-N_m k_j-1} \Omega_{c_1,k_1;\dots;c_N,k_N}
=-\frac{\pi_{\lambda_1,\dots,\lambda_N}}{2h_{\mf g}^\vee}\label{eq:N-point-series2}
\\ \qquad
{}\times\Bigg(\sum_{s\in S_N/C_N}
\frac{B\bigl(R^{\rm can}_{c_{s_1}}(\lambda_{s_1}),\dots,R^{\rm can}_{c_{s_N}}(\lambda_{s_N})\bigr)}
{\prod_{j=1}^N(\lambda_{s_j}-\lambda_{s_{j+1}})}
-\frac{\delta_{N,2}\delta_{c_1+c_2,n+1}}{2r}\frac{m_{c_1}\lambda_1^{N_m}+m_{c_2}\lambda_2^{N_m}}{(\lambda_1-\lambda_2)^2}
\Bigg),\nn
\end{gather}
where $S_N$ denotes the symmetric group and $C_N$ the cyclic group, and $s_{N+1}=s_1$.
\end{Theorem}

It also follows from~\eqref{tau09081}--\eqref{tau09082} and the fact that $ \Omega_{a,l;b,k}\in\R$
that for an arbitrary solution~$u_s({\bf t})$, $s=1,\dots,n$,
to the DS hierarchy~\eqref{DS:hier}, there exists a power series $\tau({\bf t})$, such that
\[
\frac{\p^2 \log \tau(\bt)}{\p t_{m_a+lrh} \p t_{m_b+krh}} = \Omega_{a,l;b,k}(\bt).
\]
Here $\Omega_{a,l;b,k}(\bt)$ are understood as the substitution of the solution in~$\Omega_{a,l;b,k}$.
We call $\tau({\bf t})$ the {\it tau-function of the solution $u_s({\bf t})$} to the DS hierarchy.
Note that the tau-function
$\tau({\bf t})$ is uniquely determined by the solution $u_s({\bf t})$ to the DS hierarchy~\eqref{DS:hier}
up to only a factor of the form
\begin{equation*}
\exp\biggl(d_0+\sum_{i\in E\cap \mathbb{Z}_{>0}} d_i t_i \biggr),\qquad
d_0, \, d_i\text{ are arbitrary constants}.
\end{equation*}
Clearly, for any $N\geq3$,
\[
 \frac{\partial^N \log\tau(\bt)}{\partial t_{m_{c_1}+k_1rh}\cdots\partial t_{m_{c_N}+k_Nrh}} = \Omega_{c_1,k_1;\dots;c_N,k_N} (\bt).
\]

It immediately follows from Theorem~\ref{thm:N-point} the next corollary.
\begin{Corollary}
For each $N\geq2$, let $c_1,\dots,c_N$ be arbitrarily given integers in~$\{1,\dots,n\}$.
For an arbitrary solution $u_s({\bf t})$ to the DS hierarchy~\eqref{DS:hier}, let $\tau({\bf t})$ be the
tau-function of the solution.
Then the generating series of logarithmic derivatives of $\tau({\bf t})$ $\big($replacing $\Omega_{c_1,k_1;\dots;c_N,k_N}$ in the left-hand side
of~\eqref{eq:N-point-series2} by $ \frac{\partial^N \log\tau(\bt)}{\partial t_{m_{c_1}+k_1rh}\cdots\partial t_{m_{c_N}+k_Nrh}}\big)$ is equal to
the right-hand side of~\eqref{eq:N-point-series2} with~$R_{m_a}(\lambda)$ being replaced by $R_{m_a}(\lambda;{\bf t})$.
\end{Corollary}

A further investigation of the interplay between the Hamiltonian structure~\cite{DSKV13, DS85} of the DS hierarchy
and the tau-structure $\Omega_{a,l;b,k}$ will be given in Section~\ref{sec:Ham}. In particular,
for an untwisted affine Kac--Moody algebra and the choice of the special vertex $c_0$ of its Dynkin diagram, we verify
in Theorem~\ref{thmunt} that $\Omega_{a,l;b,k}$ agree with the axioms of Hamiltonian tau-symmetry in the sense of~\cite{DLYZ16, DZ-norm}.

The paper is organized as follows. In Section~\ref{sec:matrix_res}
we apply the matrix-resolvent method to the study of tau-functions for the DS hierarchies.
In Section~\ref{sec:exa}
we apply the matrix-resolvent method to the DS hierarchies for the affine Kac--Moody
algebra $A_2^{(2)}$.
In Section~\ref{sec:Ham} we investigate relationship between
the Hamiltonian structure of the DS hierarchy
and the tau-structure.

\section[The matrix-resolvent method to tau-functions for the DS hierarchy]
{The matrix-resolvent method to tau-functions \\for the DS hierarchy}\label{sec:matrix_res}

In this section, we apply the matrix-resolvent (MR) method to the study of tau-functions for the DS hierarchies. In particular,
we will prove Theorem~\ref{thm:N-point}.

\subsection{The standard realization of affine Kac--Moody algebras}\label{sec:realizationKM}
Let $\mf g$ be a simple finite-dimensional Lie algebra of rank $n$, and let $\sigma$ be an automorphism of
$\mf g$ satisfying $\sigma^N=1$
for a positive integer~$N$.
Since $\sigma$ is diagonalizable, we have the direct sum decomposition
\begin{equation}\label{dec:sigma}
\mf g=\bigoplus_{\bar k\in\mb Z/N\mb Z}\mf g_{\bar k}
,
\end{equation}
where $\mf g_{\bar{k}}$ is the eigenspace of $\sigma$ with eigenvalue ${\rm e}^{\frac{2\pi {\rm i} k}{N}}$.

As in Section~\ref{section1}, denote by $L(\mf g)=\mf g\otimes\mb C\big(\big(\lambda^{-1}\big)\big)=\mf g\big(\big(\lambda^{-1}\big)\big)$ the space of Laurent series in the variable $\lambda^{-1}$
with coefficients in~$\mf g$.
The Lie algebra structure of $\mf g$ extends to a Lie algebra structure on $L(\mf g)$ in the obvious way.
We extend $\sigma$ to a Lie algebra homomorphism
(which we still denote by $\sigma$) $\sigma\colon L(\mf{g})\to L(\mf g)$
given by
\begin{equation}\label{sigmaonkm}
\sigma(a\otimes f(\lambda))=\sigma(a)\otimes f\big({\rm e}^{-\frac{2\pi {\rm i}}{N}}\lambda\big),
\end{equation}
for $a\in \mf g$, $f\in \mb C\big(\big(\lambda^{-1}\big)\big)$.
The subalgebra of invariant elements with respect to $\sigma$ is the twisted algebra of Laurent series in the variable $\lambda^{-1}$
with coefficients in $\mf g$, and we denote it by
\[
L(\mf g,\sigma)= L(\mf g)^\sigma
=\{a\in L(\mf{g})\mid \sigma(a)=a\}.
\]
On $L(\mf g,\sigma)$ we have the following gradation induced by the gradation \eqref{dec:sigma} of $\mf g$
and the action of $\sigma$ given by \eqref{sigmaonkm}:
\begin{equation}\label{dec:sigmatwisted}
L(\mf g,\sigma)
=\wbigoplus_{k\in \mb Z}\, L(\mf g,\sigma)_k
,
\end{equation}
where $L(\mf g,\sigma)_k=\mf g_{\bar k}\otimes \lambda^k$.

Let $r$ be the least positive integer such that $\sigma^r$ is an inner automorphism of $\mf g$. Then $r=1,2$ or $3$, and we have that $\widetilde{\mf g}$, the quotient of the affine Kac--Moody algebra $X_n^{(r)}$ by the central element $K$ (cf.~Section~\ref{section1}),
can be realized as
\[
\widetilde{\mf g}\cong L(\mf g,\sigma)
.
\]
For every vertex $c_m$ of the Dynkin diagram of $X_n^{(r)}$ there exists an automorphism $\sigma_m$ of $\mf g$
of order $N_m=ra_m$, such that $\widetilde{\mf g}\cong L(\mf g,\sigma_m)$, and the standard gradation \eqref{deg:standard}
of $\widetilde{\mf g}$ becomes the gradation of $L(\mf g,\sigma_m)$ in powers of $\lambda$ given
by \eqref{dec:sigmatwisted}, see \cite{Kac94}.
We call $L(\mf g,\sigma_m)$ the \emph{standard realization} of $\widetilde{\mf g}$ corresponding to the vertex $c_m$.

For the remaining of the section we will work with the standard realization of $\widetilde{\mf g}$ corresponding to the vertex $c_m$.

Recall from Section~\ref{section1} that $\mf a:=L(\mf g,\sigma_m)_0=\mf g_{\bar0}$ is a semisimple Lie algebra, and $e=\Lambda-e_m$ is a principal nilpotent
element. By the Jacobson--Morozov theorem~\cite{CM} there exist $\rho^\vee$ and~$f$ in~$\mf a$ such that
\begin{equation*}
\big[\rho^\vee,e\big]=e,\qquad
\big[\rho^\vee,f\big]=-f,\qquad
[e,f]=\rho^\vee.
\end{equation*}
The decomposition \eqref{dec:dynkin} of $\mf a$ is the decomposition in $\ad \rho^\vee$-eigenspaces.
Note that $\rho^\vee\in\mf a^0$, which is a Cartan subalgebra of $\mf a$.
The centralizer of $\mf a^0$ is a Cartan subalgebra of $\mf g$ (see~\cite{Kac94}). Hence,
$\rho^\vee\in\mf g$ is a semisimple element and by the representation
theory of $\mf{sl}_2$ we have the $\ad \rho^\vee$-eigenspace
decomposition of $\mf g$:
\[
\mf g=\bigoplus_{i\in\frac12\mb Z}\mf g^i,\qquad
\mf g^i=\big\{a\in\mf g\mid \big[\rho^\vee,a\big]=ia\big\}.
\]
For an eigenvector $a\in\mf g$ with respect to the adjoint action of $\rho^\vee$, we denote by $\delta(a)$ the corresponding eigenvalue,
namely $\big[\rho^\vee,a\big]=\delta(a)a$. Note that the maximal eigenvalue for the adjoint action of $\rho^\vee$ is $\frac{r(h-1)}{N_m}$.

The principal gradation \eqref{deg:principal} on $\widetilde{\mf g}\cong L(\mf g,\sigma_m)$ is then defined as follows:
if $a\otimes\lambda^k\in \widetilde{\mf g}$, $k\in\mb Z$, and $a$ is an eigenvector for $\ad\rho^\vee$, then
\begin{equation}\label{deg:principal2}
\deg^{\rm pr}\big(a\otimes\lambda^k\big)=\delta(a)+k\frac{r h}{N_m}.
\end{equation}
From equation \eqref{deg:principal2} we have that the principal gradation \eqref{deg:principal} on $\widetilde{\mf g}$ is defined by the
following linear map:
\begin{equation}\label{deg:principal3}
\ad\rho^\vee\otimes1+1\otimes\frac{rh}{N_m}\lambda\frac{\rm d}{{\rm d}\lambda}\colon\
\widetilde{\mf g}\to \widetilde{\mf g}.
\end{equation}
(In the sequel we will often omit the tensor product sign.)

As in Section~\ref{section1} we write
\begin{equation}\label{dec:principal_loop}
L(\mf g,\sigma_m)=\wbigoplus_{k\in\mb Z}\, L(\mf g,\sigma_m)^k,
\end{equation}
where elements in $\widetilde{\mf g}^k\cong L(\mf g,\sigma)^k$ have principal degree $k$.

Denote by $(\cdot\,|\,\cdot)$ the normalized invariant bilinear form on $\mf g$:
\begin{equation}\label{ckdef1}
(a|b)=\frac{1}{2h_{\mf g}^\vee}\tr(\ad a\circ\ad b),\qquad
a,b\in\mf g,
\end{equation}
where $h_{\mf g}^\vee$ is the dual Coxeter number of $\mf g$.
We extend it to a bilinear form on $L(\mf g)$ with values in $\mb C\big(\big(\lambda^{-1}\big)\big)$ by
\begin{equation}\label{ckdef2}
(a\otimes f(\lambda)|b\otimes g(\lambda))=(a|b)f(\lambda)g(\lambda),
\qquad
a,b\in\mf g,\quad
f(\lambda),g(\lambda)\in\mb C\big(\big(\lambda^{-1}\big)\big)
.
\end{equation}
Throughout the paper we will consider the restriction of this $\mb C\big(\big(\lambda^{-1}\big)\big)$-valued bilinear form to~$\widetilde{\mf g}$.

\begin{Remark}
If $a(\lambda)\in \widetilde{\mf g}\subset L(\mf g)$ and $b(\lambda)\in L(\mf g)$, then we can compute
$(a(\lambda)| b(\lambda))$.
Note that $\widetilde{\mf g}$ is not preserved by $\partial_\lambda$ if $r>1$. Nevertheless, an expression of the form
$(\partial_\lambda a(\lambda)|b(\lambda))$, $a(\lambda),b(\lambda)\in\widetilde{\mf g}$ still makes sense.
We note that, however, the operator $\lambda\partial_\lambda$ does preserve $\widetilde{\mf g}$, and one could think of
$(\partial_\lambda a(\lambda)|b(\lambda))$ as defined by $(\partial_\lambda a(\lambda)|b(\lambda))=(\lambda \partial_\lambda a(\lambda)|b(\lambda))\lambda^{-1}$.
\end{Remark}
%

\subsection{Basis of the principal Heisenberg subalgebra and basic resolvents}\label{sec:3.2}
Under the standard realization, we often write $\Lambda=\Lambda(\lambda)$, and
let us fix a basis $\{\Lambda_i(\lambda)\mid i\in E\}$ of $\mc H$ (cf.~\eqref{mcHdecompbasislambda}), with $\deg^{\rm pr}\Lambda_i(\lambda)=i$, as follows. We let
$\Lambda_1(\lambda)=\Lambda(\lambda)$
and
\begin{equation}\label{eq:normalization1}
\Lambda_{m_a+rhk}(\lambda)=\Lambda_{m_a}(\lambda)\lambda^{kN_m},\qquad
k\in\mb Z,\quad
1\leq a,b\leq n,
\end{equation}
where $\Lambda_{m_a}(\lambda)$ are normalized by the condition
\begin{equation}\label{eq:normalization}
(\Lambda_{m_a}(\lambda)|\Lambda_{m_b}(\lambda))=\delta_{a+b,n+1} h\lambda^{N_m}
,
\qquad
1\leq a,b\leq n
.
\end{equation}
Recall that $\Lambda(\lambda)=e+e_m(\lambda)$, with $e_m(\lambda)\in \widetilde{\mf g}^1$. By~\eqref{dec:sigmatwisted}
we have that $e_m(\lambda)=\tilde e_m\lambda$, for some $\tilde e_m\in\mf g_{\bar1}$.

We note that the invariance of the bilinear form \eqref{ckdef2} and the fact that $\mc H$ is abelian imply that the decomposition \eqref{dec:Lambda} is orthogonal with respect to $(\cdot\,|\,\cdot)$.
\begin{Lemma}\label{20210121:lem1b}
For $a,b=1,\dots,n$, we have
\begin{equation}\label{20210128:eq1b}
(\partial_\lambda \Lambda_{m_a}(\lambda)|\Lambda_{m_b}(\lambda))
=\delta_{a+b,n+1}\frac{m_aN_m}{r}\lambda^{N_m-1}.
\end{equation}
\end{Lemma}
\begin{proof}
Since $\deg^{\rm{pr}}\Lambda_{m_a}(\lambda)=m_a$, using the grading operator defined in \eqref{deg:principal3}, we
get the identity
\begin{gather}\label{20171207:eq1}
m_a\Lambda_{m_a}(\lambda)
=\big[\rho^\vee,\Lambda_{m_a}(\lambda)\big]+\frac{rh}{N_m}\lambda \partial_\lambda\Lambda_{m_a}(\lambda).
\end{gather}
By pairing both sides of \eqref{20171207:eq1} with $\Lambda_{m_b}(\lambda)$ we get
\begin{align*}
m_a(\Lambda_{m_a}(\lambda)|\Lambda_{m_b}(\lambda))
&=\big(\big[\rho^\vee,\Lambda_{m_a}(\lambda)\big]|\Lambda_{m_b}(\lambda)\big)
+\frac{rh}{N_m}\lambda (\partial_\lambda \Lambda_{m_a}(\lambda)|\Lambda_{m_b}(\lambda))
\\
&=\frac{rh}{N_m}\lambda (\partial_\lambda\Lambda_{m_a}(\lambda)|\Lambda_{m_b}(\lambda)),
\end{align*}
where in the last identity we used the invariance of the bilinear form and the fact that $\mc H$ is abelian.
Equation \eqref{20210128:eq1} follows by using the normalization condition given in \eqref{eq:normalization} in the LHS of the above identity.
\end{proof}

The following result will be used in Section~\ref{sec:Ham}.
\begin{Lemma}
For $a=1,\dots,n$, we have that
\begin{equation}\label{20200828:eq3}
\pi_{\mc H}(\lambda \partial_{\lambda}\Lambda_{m_a}(\lambda))=\frac{m_aN_m}{rh}\Lambda_{m_a}(\lambda),
\end{equation}
where we recall that $\pi_{\mc H}$ denotes the projection onto~$\mc H$ with respect to~\eqref{dec:Lambda}.
\end{Lemma}
\begin{proof}
From~\eqref{mcHdecompbasislambda} we have that
$\pi_{\mc H}(\lambda \partial_{\lambda}\Lambda_{m_a}(\lambda))=\sum_{j\in E}c_j\Lambda_j(\lambda)$.
Using \eqref{exponents_set} and equations \eqref{eq:normalization1}--\eqref{eq:normalization} we get
\begin{align}
\big(\pi_{\mc H}(\lambda \partial_{\lambda}\Lambda_{m_a}(\lambda))|\Lambda_{m_b+rhk}(\lambda)\big)
&=\sum_{i=1,\dots,n,\,l\in\mb Z}\delta_{i+b,n+1}c_{i,l}h\lambda^{(l+k+1)N_m}\nonumber
\\
&=\sum_{l\in\mb Z}c_{n+1-b,l}h\lambda^{(l+k+1)N_m}.\label{eq:lemma1}
\end{align}
Hence, since the decomposition \eqref{dec:Lambda} is orthogonal with respect to $(\cdot\,|\,\cdot)$, from equations \eqref{eq:normalization1}--\eqref{eq:normalization} and
\eqref{20210128:eq1b} we have that
\begin{align}
\big(\pi_{\mc H}(\lambda \partial_{\lambda}\Lambda_{m_a}(\lambda))|\Lambda_{m_b+rhk}(\lambda)\big)
&=(\lambda\partial_{\lambda}\Lambda_{m_a}(\lambda)|\Lambda_{m_b+rhk}(\lambda))\nonumber
\\
&=\delta_{a+b,n+1}\frac{m_aN_m}{r}\lambda^{(k+1)N_m}.\label{eq:lemma2}
\end{align}
Combining equations \eqref{eq:lemma1} and \eqref{eq:lemma2} it follows that
$c_{i,l}=\delta_{i,a}\delta_{l,0}\frac{m_aN_m}{rh}$, for $i=1,\dots,n$ and $l\in\mb Z$ thus proving equation \eqref{20200828:eq3}.
\end{proof}

Recall from Section~\ref{section1} the definition (cf.~Definition~\ref{defresolvents}) of the resolvents $R_i$, $i\in E$.
Under the standard realization of~$\widetilde{\mf g}$, we will often write $R_i=R_i(\lambda)$.
Using the normalization \eqref{eq:normalization} and the invariance of the $\mb C\big(\big(\lambda^{-1}\big)\big)$-valued
bilinear form on~$\widetilde{\mf g}$ we get
\begin{equation}\label{20171106:eq1}
(R_{m_a}(\lambda)|R_{m_b}(\lambda))=(\Lambda_{m_a}(\lambda)|\Lambda_{m_b}(\lambda))=\delta_{a+b,n+1} h \lambda^{N_m}
,
\qquad
1\leq a,b\leq n
.
\end{equation}
For every $a=1,\dots,n$, we decompose $R_a(\lambda)$ according to~\eqref{dec:sigmatwisted} as follows:
\begin{equation}\label{20210930:eq1}
R_a(\lambda)=\sum_{k\in\mb Z}R_{a;k}(\lambda)
,
\end{equation}
where $R_{a;k}(\lambda)\in\mf g_{\bar k}\otimes\lambda^k$. On the other hand, by \eqref{eq:R_a} we have that
\begin{equation}\label{eq:R_a-principal}
R_a(\lambda) =\Lambda_{m_a}(\lambda)+\text{lower order terms},
\end{equation}
where lower order terms are considered with respect to the principal gradation~\eqref{dec:principal_loop}.

\begin{Lemma}
For $a,b=1,\dots,n$, we have
\begin{equation}\label{20210128:eq1}
(\partial_\lambda R_a(\lambda)|R_b(\lambda))
=\delta_{a+b,n+1}\frac{m_aN_m}{r}\lambda^{N_m-1}
.
\end{equation}
\end{Lemma}
\begin{proof}
It is immediate to check, using equation~\eqref{LRcomm}, the invariance of the bilinear form $(\cdot|\cdot)$
and the fact that $[R_{m_a}(\lambda),R_{m_b}(\lambda)]=0$, that $\frac{\partial}{\partial_x}(\partial_\lambda R_{m_a}(\lambda)|R_{m_b}(\lambda))=0$.
Hence,
\[
(\partial_\lambda R_{m_a}(\lambda)|R_{m_b}(\lambda))=(\partial_\lambda\Lambda_{m_a}(\lambda)|\Lambda_{m_b}(\lambda)).
\]
The claim follows from Lemma \ref{20210121:lem1b}.
\end{proof}
%

\subsection{From basic resolvents to tau-function}\label{section33}
Using the basis of $\mc H$ given by \eqref{eq:normalization1} and the fact that the standard gradation
of $\widetilde{\mf g}$ corresponds to the gradation of $L(\mf g,\sigma_m)$ in powers of $\lambda$, we write the pre-DS hierarchy \eqref{eq:preDS}
as
\begin{equation}\label{eq:preDS2}
\frac{\partial\mc L}{\partial t_{m_a+rhk}}= \big[\big(\lambda^{kN_m}R_{m_a}(\lambda)\big)_+,\mc L\big],\qquad
a=1,\dots,n, \quad k\in\mb Z_+,
\end{equation}
where the subscript $+$ stands for the polynomial part in~$\lambda$
(we are choosing $i=m_a+rhk\in E\cap \mathbb{Z}_{>0}$).
Let
$\pi_\lambda\colon \mb C\big(\big(\lambda^{-1}\big)\big)\to \lambda^{-1}\mb C\big[\big[\lambda^{-N_m}\big]\big]$ be the linear map defined via
\[
\lambda^k\mapsto
\begin{cases}
\lambda^k & \text{if}\  k\equiv -1 \, (\bmod N_m),\ k<0,
\\
0 & \text{otherwise},
\end{cases}
\]
where $k\in\mb Z$.
Let $\epsilon$ be a primitive $N_m$-root of unity. Recall that, for $h\in\mb Z$, we have
\[
\sum_{k=0}^{N_m-1}\epsilon^{kh}=
\begin{cases}
N_{m} & \text{if}\  h\equiv0 \, (\bmod N_m),
\\
0 &\text{otherwise}.
\end{cases}
\]
Note that $\pi_\lambda$ can be equivalently defined as follows:
\begin{equation}\label{eq:pi}
\pi_\lambda\big( f(\lambda) \big)=\frac{1}{N_m}\sum_{k=0}^{N_m-1}\epsilon^kf(\epsilon^k\lambda)_-,
\qquad
f(\lambda)\in\mb C\big(\big(\lambda^{-1}\big)\big),
\end{equation}
where $f(\lambda)_-\in\lambda^{-1}\mb C\big[\big[\lambda^{-1}\big]\big]$ denotes the singular part of $f(\lambda)$.
Similarly, we will denote
\[
\pi_{\lambda,\mu}=\pi_\lambda\circ\pi_\mu\colon\ \mb C\big(\big(\lambda^{-1},\mu^{-1}\big)\big)\to\lambda^{-1}\mu^{-1}\mb C\big[\big[\lambda^{-N_{m}},\mu^{-N_m}\big]\big]
\]
and
\[
\pi_{\lambda,\mu,\eta}=\pi_\lambda\circ\pi_\mu\circ\pi_\eta\colon\ \mb C\big(\big(\lambda^{-1},\mu^{-1},\eta^{-1}\big)\big)\to\lambda^{-1}\mu^{-1}\eta^{-1}\mb C\big[\big[\lambda^{-N_{m}},\mu^{-N_m},\eta^{-N_m}\big]\big].
\]
Clearly, the maps $\pi_\lambda$, $\pi_\mu$ and $\pi_\eta$ commute.
We extend $\pi_\lambda$ to a map $\mc A^q\otimes\mb C\big(\big(\lambda^{-1}\big)\big)\to\mc A^q\otimes\lambda^{-1}\mb C\big(\big(\lambda^{-N_m}\big)\big)$
in the obvious way.
By using equations \eqref{20171106:eq1}, \eqref{20210128:eq1} and \eqref{20171106:eq2} we find
\[
(\lambda-\mu)^{-2}\bigg(\big(R_{m_a}(\lambda)|R_{m_b}(\mu)\big) -\frac{\delta_{a+b,n+1}}{r}\big(m_a\lambda^{N_m}+m_b\mu^{N_m}\big)\bigg)
\in\mc A^q\big(\big(\lambda^{-1},\mu^{-1}\big)\big).
\]
Hence, the LHS of \eqref{eq:two-point-series} is well defined; in other words, $\Omega_{a,k;b,l}\in\mc A^q$ are well defined (see~\eqref{eq:Fab}).
Recall also from Section~\ref{section1} that $\Omega_{a,k;b,l}$ are gauge invariant, hence they actually belong to $\mc R\subset\mc A^q$.

For a Laurent series $a(\lambda)=\sum_{i\leq M} a_i\lambda^i$, we denote
$\Res_{\lambda}a(\lambda)=a_{-1}$ (which is equal to $-\Res_{\lambda=\infty}a(\lambda)$).
The following result extends Proposition~2.3.2 in \cite{BDY21} to our current more general setting.
\begin{Lemma}
For $a,b=1,\dots,n$, we have
\begin{equation}\label{20171113:eq1}
\Res_\mu F_{a,b}(\lambda,\mu)=\pi_{\lambda}(R_{m_a}(\lambda)|\partial_\lambda R_{m_b}(\lambda)_+)
.
\end{equation}
In particular, for every $a=1,\dots,n$, we have
\begin{equation}\label{20171113:eq2}
\Res_\mu F_{a,1}(\lambda,\mu)=\pi_{\lambda}(R_{m_a}(\lambda)|\tilde e_m).
\end{equation}
\end{Lemma}
\begin{proof}
Equation \eqref{20171113:eq1} follows by taking the residue in $\mu$ in both sides of equation \eqref{eq:two-point-series}
and using the identity (which holds for an arbitrary Laurent series $a(\mu)$)
\[
\Res_\mu a(z)\iota_\mu(\mu-\lambda)^{-2}=\partial_\lambda a(\lambda)_+,
\]
where $\iota_\mu$ denotes the expansion in the region $|\mu|>|\lambda|$.
Equation \eqref{20171113:eq2} is obtained from \eqref{20171113:eq1} by recalling that
$R_1(\lambda)_+=\Lambda(\lambda)=e+\tilde e_m\lambda$.
\end{proof}

For simplicity of notation, let us denote $G_a(\lambda)=\pi_\lambda (R_{m_a}(\lambda)|\tilde e_m)$, $a=1,\dots,n$.
Using \eqref{eq:Fab} and \eqref{20171113:eq2}, we have
\begin{equation}\label{eq:G}
G_a(\lambda)=\sum_{k\in\mb Z_{\geq0}}\Omega_{a,k;1,0}\lambda^{-N_mk-1}\in\mc \R\big[\big[\lambda^{-N_m}\big]\big]\lambda^{-1}.
\end{equation}
It follows from the relations~\eqref{tau09082} with $c=1$, $m=0$ and the identification $x=-t_1$
 that $\Omega_{a,k;1,0}$ are densities of conservations laws for the DS hierarchy \eqref{DS:hier}.
The relationship between the series $G_a(\lambda)$ and the Hamiltonian structure of the DS hierarchies will be studied in Section~\ref{sec:Ham}.

Let us introduce the
{\it loop operator} for
the pre-DS hierarchy~\eqref{eq:preDS2} as follows:
\begin{equation*}
\nabla_a(\lambda)=\sum_{k\in\mb Z_{\geq0}}\lambda^{-N_m k-1}\frac{\partial}{\partial t_{m_a+rhk}}, \qquad a=1,\dots,n.
\end{equation*}
\begin{Lemma}\label{20171104:lem1}
For every $a=1,\dots,n$, we have
\begin{equation}\label{20171104:eq1}
\nabla_a(\lambda)R_{m_b}(\mu)=\pi_\lambda\frac{[R_{m_a}(\lambda),R_{m_b}(\mu)]}{\lambda-\mu}.
\end{equation}
\end{Lemma}
\begin{proof}
We have
\begin{align}
\nabla_a(\lambda)R_{m_b}(\mu)
&=\sum_{k\in\mb Z_{\geq0}}\frac{\partial R_{m_b}(\mu)}{\partial t_{m_a+rhk}}\lambda^{-N_mk-1}
=\sum_{k\in\mb Z_{\geq0}} \big[\big(\mu^{N_mk}R_{m_a}(\mu)\big)_+,R_{m_b}(\mu)\big]\lambda^{-N_mk-1}\nonumber
\\
&=\sum_{k\in\mb Z_{\geq0}}
\Res_\rho \big[\rho^{N_mk}R_{m_a}(\rho),R_{m_b}(\mu)\big]\lambda^{-N_mk-1}\iota_\rho(\rho-\mu)^{-1}.
\label{20171104:eq3}
\end{align}
In the second identity we used equation \eqref{20171104:eq2} and in the third identity we used the fact that
\[
a(w)_+=\Res_z a(z)\iota_z(z-w)^{-1},
\]
which holds for any formal series $a(z)$ (here $\Res_z$ is the coefficient of $z^{-1}$ and $\iota_z$ is the expansion in the region
$|z|>|w|$).
Note that
\begin{equation}\label{20171104:eq4}
\sum_{k\in\mb Z_\geq0}
\rho^{N_mk}\lambda^{-N_mk-1}=\lambda^{N_m-1}\iota_\rho\big(\lambda^{N_m}-\rho^{N_m}\big)^{-1}.
\end{equation}
Using equation \eqref{20171104:eq4}, we can rewrite \eqref{20171104:eq3} as
{\samepage\begin{align}
\nabla_a(\lambda)R_{m_b}(\mu)
&=\frac{\lambda^{N_m-1}}{2\pi {\rm i}}\oint_{|\mu|<|\rho|<|\lambda|}
\frac{[R_{m_a}(\rho),R_{m_b}(\mu)]}{(\lambda^{N_m}-\rho^{N_m})(\rho-\mu)}\nonumber
\\
&=\lambda^{N_m-1}
\Bigg(\sum_{k=0}^{N_m-1}\frac{[R_{m_a}(\epsilon^k\lambda),R_{m_b}(\mu)]} {N_m(\epsilon^k\lambda)^{N_m-1}(\epsilon^k\lambda-\mu)}
-\frac{[R_{m_a;N_m}(\lambda),R_{m_b}(\mu)]}{\lambda^{N_m}}\Bigg)\nonumber
\\
&=\frac{1}{N_m}\sum_{k=0}^{N_m-1}\epsilon^k\frac{[R_{m_a}(\epsilon^k\lambda),R_{m_b}(\mu)]} {(\epsilon^k\lambda-\mu)}
-\frac{[R_{m_a;N_m}(\lambda),R_{m_b}(\mu)]}{\lambda},
\label{20171104:eq5}
\end{align}
where the second identity follows by the residue theorem.}
Equation \eqref{20171104:eq1} is obtained by combining equations \eqref{20171104:eq5} and \eqref{eq:pi} and the fact that (cf.~\eqref{20210930:eq1})
\[
\frac{1}{N}\sum_{k=0}^{N_m-1}\epsilon^k \bigg(\frac{[R_{m_a}(\epsilon^k\lambda),R_{m_b}(\mu)]}{(\epsilon^k\lambda-\mu)}\bigg)_+
=\frac{[R_{m_a;N_m}(\lambda),R_{m_b}(\mu)]}{\lambda},
\]
where $(\cdot)_+$ denotes the non-negative part in powers of $\lambda$. This concludes the proof.
\end{proof}

Recall the definitions of the differential polynomials $\Omega_{c_1,k_1;\dots;c_N,k_N}$ from \eqref{eq:Fab} and~\eqref{defOmegaN}.
We~have the following proposition.

\begin{Proposition}\label{prop:N-point}
For each $N\geq2$, let $c_1,\dots,c_N$ be arbitrarily given integers in~$\{1,\dots,n\}$.
We have
\begin{gather}
\label{eq:N-point-series2-Om}
\sum_{k_1,\dots,k_N\in\mb Z_{\geq0}} \prod_{j=1}^N
\lambda_j^{-N_m k_j-1} \Omega_{c_1,k_1;\dots;c_N,k_N}
=-\frac{\pi_{\lambda_1,\dots,\lambda_N}}{2h_{\mf g}^\vee}
\\ \qquad
{}\times\Bigg(\sum_{s\in S_N/C_N}\!\!\!\!
\frac{B\bigl(R_{c_{s_1}}(\lambda_{s_1};{\bf t}),\dots,R_{c_{s_N}}(\lambda_{s_N};{\bf t})\bigr)}
{\prod_{j=1}^N(\lambda_{s_j}-\lambda_{s_{j+1}})} -\frac{\delta_{N,2}\delta_{c_1+c_2,n+1}}{2r}\frac{m_{c_1}\lambda_1^{N_m}+m_{c_2}\lambda_2^{N_m}} {(\lambda_1-\lambda_2)^2}\Bigg),\nn
\end{gather}
where $S_N$ denotes the symmetric group and $C_N$ the cyclic group, and $s_{N+1}=s_1$.
\end{Proposition}
\begin{proof}
With the above Lemma~\ref{20171104:lem1}, the proof could
now follow the same lines as the proof of the KdV case~\cite{BDY16} (cf.~also~\cite{BDY21,DYZ21}); so we omit the details here.
\end{proof}

The properties \eqref{tau09081}, \eqref{tau09082} follow from the symmetric nature
of identities with $N=2$ and $N=3$ of Proposition~\ref{prop:N-point}, respectively.

\begin{proof}[Proof of Theorem~\ref{thm:N-point}]
The theorem follows from~\eqref{eq:N-point-series2-Om}, equation~\eqref{RcanNR}
and the invariance of the multilinear form~$B$.
\end{proof}

\begin{Remark}
DS hierarchies associated to untwisted affine Kac--Moody algebras and their tau functions
were intensively studied via various methods \cite{BFORFW, BDY21, CW19, DLZ,
FGM, HM, KW, Wu17}.
In particular, using the matrix-resolvent method, gauge invariant differential polynomials $\Omega_{a,l;b,k}$
satisfying the tau-structure properties \eqref{tau09081} and \eqref{tau09082} (hence leading to
the construction of tau-function following the scheme of~\cite{DZ-norm})
were defined in~\cite{BDY21} (cf.~also references therein for earlier results).
In the twisted cases, these were also constructed in \cite{LWZ19,LWZZ20} with a different method.
It would be interesting to extend the matrix-resolvent method to the {\it generalized
Drinfeld--Sokolov hierarchies} \cite{dGHM, DSJKV21, DSKV13, Dinar2, Dinar3, HM}.
\end{Remark}

\section{Examples}\label{sec:exa}
In this section, we apply the matrix-resolvent construction to the
Drinfeld--Sokolov hierarchies associated to the affine Kac--Moody algebra $A_2^{(2)}$.

The Dynkin diagram of the affine Kac--Moody algebra $A_2^{(2)}$ is
\begin{equation}\label{diagram}
\begin{tikzpicture}[start chain]
\dnodeanj{}{0}
\dnodeanj{}{1}
\path (chain-1) -- node {\QLeftarrow} (chain-2);
\end{tikzpicture}
\end{equation}
and its Cartan matrix is
\begin{equation}\label{eq:Cartan}
C=\begin{pmatrix}
\hphantom{-}2&-4\\-1&\hphantom{-}2
\end{pmatrix}
.
\end{equation}
As discussed in Section~\ref{section1} we have a set of Chevalley generators $e_0$, $e_1$, $h_0$, $h_1$, $f_0$, $f_1$ satisfying
the relations \eqref{rel:Lie} and \eqref{rel:Serre}. It follows immediately from the definition \eqref{eq:Cartan} of $C$ that
\[
a_0=a_1^\vee=2,\qquad a_1=a_0^\vee=1,
\]
hence the Coxeter number and dual Coxeter number of $A_2^{(2)}$ are $h=h^\vee=3$. The set of exponents $E$ has the form (cf.~\eqref{exponents_set})
\[
E=\left(1+6\mb Z\right)\cup\left(5+6\mb Z\right),
\]
that is we have $m_1=1$ and $m_2=5$ in \eqref{exponents_set}.

Let $K=h_0+2h_1$ be the central element. Recall that we are interested in the quotient of $A_2^{(2)}$ by the one-dimensional space generated by $K$, which we denote by $\widetilde{\mf g}$. We describe this quotient as a subspace of $L(\mf{sl}_3)$ following the discussion in Section~\ref{sec:realizationKM}. The normalized invariant bilinear form~\eqref{ckdef1} on $\mf{sl}_3$ is the trace form which we extend to $L(\mf{sl}_3)$ in the natural way.

\subsection{The Sawada--Kotera hierarchy}\label{sec:SKhierarchy}
The Drinfeld--Sokolov hierarchy associated to $A_2^{(2)}$ and the choice of the vertex~$c_0$ of its Dynkin diagram~\eqref{diagram} is known to be
 the Sawada--Kotera hierarchy \cite{SK74}. Following Section~\ref{sec:matrix_res} we compute the basic matrix resolvents for this hierarchy.

\subsubsection[Principal and standard gradations for A\_2\textasciicircum{}\{(2)\} and the c\_0 vertex]{Principal and standard gradations for $\boldsymbol{A_2^{(2)}}$
and the $\boldsymbol{c_0}$ vertex}

In this case there exists an automorphism $\sigma_0$ of $\mf{sl}_3$ of order $N_0=4$
such that $\widetilde{\mf g}=L(\mf{sl}_3,\sigma_0)\subset L(\mf{sl}_3)$ (cf.~Section~\ref{sec:realizationKM}). Explicitly,
\[
\widetilde{\mf g}=\left\{
\left.
\begin{pmatrix}
a_1(\lambda)+a_2(\lambda)&b_1(\lambda)+b_2(\lambda)&p(\lambda)\\
c_1(\lambda)+c_2(\lambda)&-2a_2(\lambda)&b_1(\lambda)-b_2(\lambda)\\
r(\lambda)&c_1(\lambda)-c_2(\lambda)&a_2(\lambda)-a_1(\lambda)
\end{pmatrix}
\right|
\begin{array}{l}
a_1(\lambda),p(\lambda),r(\lambda)\in\mb C((\lambda^{-4}))\\
b_2(\lambda),c_1(\lambda)\in\mb C((\lambda^{-4}))\lambda \\
a_2(\lambda)\in\mb C((\lambda^{-4}))\lambda^2\\
b_1(\lambda),c_2(\lambda)\in\mb C((\lambda^{-4}))\lambda^{3}
\end{array}
\right\}\!.
\]
Let us consider the following Chevalley generators for $\widetilde{\mf g}$ ($E_{ij}$ denotes
the elementary matrix):
\begin{gather}
e_0(\lambda)=(E_{21}+E_{32})\lambda,\qquad
h_0=-2(E_{11}-E_{33}),\qquad
f_0(\lambda)=2(E_{12}+E_{23})\lambda^{-1},\nn
\\
e_1=E_{13},\qquad
h_1=E_{11}-E_{33},\qquad
f_1=E_{31}.
\label{gen_c0}
\end{gather}
The {\it principal gradation} is defined by the linear map \eqref{deg:principal3}, where $\rho^\vee=h_1/2$. Explicitly, we have
\begin{align*}
&\widetilde{\mf g}^{6k}=\mb C h_1\lambda^{4k},\qquad
\widetilde{\mf g}^{6k+1}=\mb C e_0(\lambda)\lambda^{4k}\oplus\mb C e_1\lambda^{4k},\qquad
\widetilde{\mf g}^{6k+2}=\mb C (E_{12}-E_{23})\lambda^{4k+1},
\\
&\widetilde{\mf g}^{6k+3}=\mb C (E_{11}-2E_{22}+E_{33})\lambda^{4k+2},\qquad
\widetilde{\mf g}^{6k+4}=\mb C (E_{21}-E_{32})\lambda^{4k+3},
\\
&\widetilde{\mf g}^{6k+5}=\mb C f_0(\lambda)\lambda^{4k+4}\oplus\mb C f_1\lambda^{4k+4},
\end{align*}
where $k\in\mb Z$.
The {\it standard gradation} corresponding to~$c_0$ is the gradation in powers of $\lambda$, i.e.,
\begin{align*}
&\widetilde{\mf g}_{4k}=\mb C f_1\lambda^{4k}\oplus\mb C h_1\lambda^{4k}\oplus\mb C e_1\lambda^{4k},
\qquad
\widetilde{\mf g}_{4k+1}=\mb C e_0(\lambda)\lambda^{4k}\oplus\mb C(E_{12}-E_{23})\lambda^{4k+1},
\\
&\widetilde{\mf g}_{4k+2}=\mb C (E_{11}-2E_{22}+E_{33})\lambda^{4k+2},
\qquad
\widetilde{\mf g}_{4k+3}=\mb C f_0(\lambda)\lambda^{4k+4}\oplus\mb C(E_{21}-E_{32})\lambda^{4k+3},
\end{align*}
where $k\in\mb Z$.
Note that $\mf a=\widetilde{\mf g}_0\cong\mf{sl}_2$,
moreover, $\mf n = \mb Cf_1 \subset \mf b= \mb C f_1 \oplus \mb C h_1$.

Recall that the element $\Lambda(\lambda)=e_0(\lambda)+e_1\in\widetilde{\mf g}$ is semisimple
and we have the direct sum decomposition~\eqref{dec:Lambda}.
Let, as in Section~\ref{section1}, $\mc H=\Ker \ad\Lambda(\lambda)$. It is immediate to check that
\begin{equation}\label{eq:Lambda-c0}
\mc H=\Big(\widehat\bigoplus_{k\in\mb Z}\mb C\Lambda(\lambda)\lambda^{4k}\Big)\oplus
\Big(\widehat\bigoplus_{k\in\mb Z}\mb C(f_0(\lambda)+2f_1)\lambda^{4k}\Big).
\end{equation}
We rewrite~\eqref{eq:Lambda-c0} as in~\eqref{mcHdecompbasislambda} using the following basis $\{\Lambda_i(\lambda)\mid i\in E\}$:
\begin{gather*}
\Lambda_1(\lambda)=\Lambda(\lambda),\qquad
\Lambda_5(\lambda)=\bigg(\frac12f_0(\lambda)+f_1\bigg)\lambda^4,
\\
\Lambda_{1+6k}(\lambda)=\Lambda_1(\lambda)\lambda^{4k},\qquad
\Lambda_{5+6k}(\lambda)=\Lambda_5(\lambda)\lambda^{4k}.
\end{gather*}
Here $k\in\mb Z$.
This basis satisfies the normalization conditions \eqref{eq:normalization1} and \eqref{eq:normalization}.

\subsubsection{The matrix resolvent}
Take the DS gauge $V=\mb Cf_1$ (cf.~\eqref{bVen}). The element $\mc L^{\text{can}}$ in~\eqref{Lcan} has the form
\[
\mc L^{\text{can}}=\partial+\Lambda(\lambda)+uf_1
=\begin{pmatrix}
\partial&0&1\\
\lambda&\partial&0\\
u&\lambda&\partial
\end{pmatrix}\!.
\]
We have that $\R=\mb C[u,u_x,u_{2x},\dots]$ is the algebra of differential polynomials
in~$u$.
(Here and in what follows, for a smooth function $y=y(x)$ of $x$, we denote $y_{nx}:=\p_x^n(y)$, $n\geq0$.)
Let
\begin{equation}\label{eq:R-SK}
R(\lambda)=\begin{pmatrix}
a_1(\lambda)+a_2(\lambda)&b_1(\lambda)+b_2(\lambda)&p(\lambda)\\
c_1(\lambda)+c_2(\lambda)&-2a_2(\lambda)&b_1(\lambda)-b_2(\lambda)\\
r(\lambda)&c_1(\lambda)-c_2(\lambda)&a_2(\lambda)-a_1(\lambda)
\end{pmatrix}
\in \R\otimes \widetilde{\mf g}
\end{equation}
be a resolvent of~$\mc L^{\text{can}}$. Then
$[\mc L^{\text{can}},R(\lambda)]=0$; cf.~\eqref{LRcomm}.
It follows that
\begin{align*}
&a_1(\lambda)=\frac12 p_x(\lambda),\qquad a_2(\lambda)=\frac13 (ub_1(\lambda)-b_{1,2x}(\lambda))\lambda^{-1},\\
&b_2(\lambda)=\frac13(ub_{1,x}(\lambda)+u_x b_1(\lambda)-b_{1,3x}(\lambda))\lambda^{-2},
\\
&c_1(\lambda)=p(\lambda)\lambda -\frac13(2b_{1,x}(\lambda)u_x+ub_{1,2x}(\lambda)+b_1(\lambda)u_{2x}-b_{1,4x}(\lambda))\lambda^{-2},\\
& c_2(\lambda)=b_{1,x}(\lambda), \qquad r(\lambda)=b_1(\lambda)\lambda+up(\lambda)-\frac12 p_{2x}(\lambda)
,
\end{align*}
where
$b_1=b_1(\lambda)\in\R\big(\big(\lambda^{-4}\big)\big)\lambda^3$ and $p=p(\lambda)\in\R((\lambda^{-4}))$ satisfy the following system of ODEs:
\begin{gather}\label{eq:system1}
\begin{cases}
9p_x\lambda^3+2(uu_x-u_{3x})b_1+2\big(u^2-3u_{2x}\big)b_{1,x}-6u_x b_{1,2x}-4u b_{1,3x}+2b_{1,5x}=0,
\\
6\lambda b_{1,x}+2u_x p+4up_x-p_{3x}=0.
\end{cases}
\end{gather}
To find the basic resolvent $R_1(\lambda)$, we write it as in~\eqref{eq:R_a-principal} as follows:\vspace{-1mm}
\[
R_1(\lambda)=\Lambda(\lambda)+\text{ terms of lower degree} .
\]
Then one can solve~\eqref{eq:system1} for $b_1(\lambda)=\sum_{k\geq0}b_{1,k}\lambda^{-4k-1}$ and $p(\lambda)=1+\sum_{k\geq1}p_{k}\lambda^{-4k}$
recursively. The first few terms are given by\vspace{-1mm}
\begin{gather}
b_1(\lambda)=-\frac u3 \lambda^{-1}
+\frac1{243}\big({-}7u^4+42u^2u_{2x}+21u u_{x}^2-21uu_{4x}-21u_{2x}^2-21u_{x}u_{3x}+3u_{6x}\big)\lambda^{-5}\nn
\\ \hphantom{b_1(\lambda)=}
{}+\cdots,\nn
\\
p(\lambda)=1+\frac{1}{81}\big(4u^3-9u_{x}^2-18uu_{2x}+6u_{4x}\big)\lambda^{-4}
+\frac1{6561}\big({-}18 u_{10x}+162 uu_{8x}\notag
\\ \hphantom{p(\lambda)=}
{}-522 u^2u_{6x}+798 u^3u_{4x}+1395 u_{4x}^2-3456 u u_{3x}^2-630 u_{2x}u^4+3591 u^2 u_{2x}^2\notag
\\ \hphantom{b_1(\lambda)=}
{}-3252 u_{2x}^3-1260 u^3 u_{x}^2+1134 u_{x}^4+648 u_{x}u_{7x}+1548 u_{2x}u_{6x}+2376 u_{3x}u_{5x}
\notag
\\ \hphantom{b_1(\lambda)=}
{}-3132 u u_{x}u_{5x}-6066 u u_{2x}u_{4x}-4428 u_{x}^2u_{4x}+4788 u^2 u_{x}u_{3x}
+9324 u u_{x}^2 u_{2x}\nn
\\ \hphantom{b_1(\lambda)=}
{}-14184 u_{x} u_{2x}u_{3x}+35 u^6\big)\lambda^{-8}+\cdots.
\label{b1SK}
\end{gather}
Similarly, to find the basic resolvent $R_5(\lambda)$, we write it as in~\eqref{eq:R_a-principal} as follows:
\[
R_5(\lambda)=\Lambda_5(\lambda)+\text{terms of lower degree} .
\]
Then one can solve~\eqref{eq:system1} for $b_1(\lambda)=\lambda^3+\sum_{k\geq0}b_{1,k}\lambda^{-4k-1}$
and $p(\lambda)=\sum_{k\geq0}p_{k}\lambda^{-4k}$ recursively. The first few terms are
\begin{gather}
b_1(\lambda)=\lambda^{3}
+\frac1{81}\big(5u^3-15uu_{2x}+3u_{4x}\big)\lambda^{-1}
+\frac{1}{6561}\big({-}9 u_{10x}+99 uu_{8x}-396 u^2u_{6x}\notag
\\ \hphantom{b_1(\lambda)=}
{}+726 u^3u_{4x}
+693 u_{4x}^2-2079 uu_{3x}^2-660 u^4 u_{2x}+2772 u^2 u_{2x}^2-1716 u_{2x}^3\notag
\\ \hphantom{b_1(\lambda)=}
{}-990 u^3 u_{x}^2+297 u_{x}^4+297 u_{x}u_{7x}+792 u_{2x}u_{6x}+1188 u_{3x}u_{5x}-1782 uu_{x }u_{5x}\notag
\\ \hphantom{b_1(\lambda)=}
{}-3762 uu_{2x}u_{4x}-1782 u_{x}^2u_{4x}+3366 u^2u_{x}u_{3x}
+4950 u u_{x}^2 u_{2x}-6732 u_{x}u_{2x}u_{3x}\notag
\\ \hphantom{b_1(\lambda)=}
{}+44 u^6\big)\lambda^{-5}+\cdots,\notag
\\
p(\lambda)=\frac19\big(2u_{2x}-u^2\big)
+\frac{1}{729}\big({-}6 u_{8x}+42 uu_{6x}-96 u^2u_{4x}+117 u_{3x}^2
+100 u^3 u_{2x}\notag
\\ \hphantom{b_1(\lambda)=}
{}-288 u u_{2x}^2+150 u^2 u_{x}^2+126 u_{x}u_{5x}+222 u_{2x}u_{4x}
-384 uu_{x}u_{3x}\nn
\\ \hphantom{b_1(\lambda)=}
{}-396 u_{x}^2 u_{2x}-8 u^5\big)\lambda^{-4}+\cdots.
\label{b1SK2}
\end{gather}
Recall from Section~\ref{sec:3.2} that $\Lambda(\lambda)=e+e_0(\lambda)$, where $e=e_1$ and
$e_0(\lambda)=\tilde e_0\lambda$. From~\eqref{gen_c0} we
have that $\tilde e_0=E_{21}+E_{32}$. Let $R(\lambda)$ be as in~\eqref{eq:R-SK}, then $(R(\lambda)|\tilde e_0)=2b_1(\lambda)$.
Hence, recalling the definition of the series $G_a(\lambda)$, $a=1,2$, given in~\eqref{eq:G} and using
\eqref{b1SK}--\eqref{b1SK2} we have the following expression for the first few terms of the
tau-structure of the DS hierarchy
\begin{gather}
\Omega_{1,0;1,0}=-\frac{2}{3}u,\nn
\\
\Omega_{1,1;1,0}=\frac2{243}\big({-}7u^4+42u^2u_{2x}+21uu_{x}^2-21uu_{4x}-21u_{2x}^2-21u_{x}u_{3x}+3u_{6x}\big),\nn
\\
\Omega_{2,0;1,0}=\frac2{81}\big(5u^3-15uu_{2x}+3u_{4x}\big),\nn
\\
\Omega_{2,1;1,0}=\frac{2}{6561}\big({-}9 u_{10x}+99 uu_{8x}-396 u^2u_{6x}
+726 u^3u_{4x}+693 u_{4x}^2-2079 uu_{3x}^2\nn
\\ \hphantom{\Omega_{2,1;1,0}=}
{}-660 u^4 u_{2x}+2772 u^2 u_{2x}^2-1716 u_{2x}^3 -990 u^3 u_{x}^2+297 u_{x}^4
+297 u_{x}u_{7x}\nn
\\ \hphantom{\Omega_{2,1;1,0}=}
{}+792 u_{2x}u_{6x}+1188 u_{3x}u_{5x}
-1782 uu_{x}u_{5x}-3762 uu_{2x}u_{4x}-1782 u_{x}^2u_{4x}\nn
\\ \hphantom{\Omega_{2,1;1,0}=}
{}+3366 u^2u_{x}u_{3x}+4950 u u_{x}^2 u_{2x}-6732 u_{x}u_{x}u_{3x}+44 u^6\big).\label{tau:SK}
\end{gather}
Taking $b=c=1$ and $k=m=0$ in \eqref{tau09082}, and using the first equation in \eqref{tau:SK}
and the fact that $\partial_x=-\partial_{t_1}$, the
DS hierarchy for $A_2^{(2)}$ and~$c_0$ can be written as
\begin{equation}\label{latestchange}
\frac{\p u}{\p t_{m_a+6l}} = \frac32\partial_x \Omega_{a,l;1,0},\qquad
a=1,2, \quad l\geq0.
\end{equation}
From \eqref{tau:SK} and \eqref{latestchange} we get the following first few equations of the hierarchy:
\begin{gather*}
\frac{\partial u}{\partial t_1}=-u_{x},
\\
\frac{\partial u}{\partial t_7}=
\frac{1}{27} u_{7x}-\frac{7}{27} uu_{5x}+\frac{14}{27} u^2u_{3x}
-\frac{28}{81} u^3 u_{x}+\frac{7}{27} u_{x}^3-\frac{14}{27} u_{x}u_{4x}
-\frac{7}{9} u_{2x}u_{3x}+\frac{14}{9} u u_{x}u_{2x},
\\
\frac{\partial u}{\partial t_5}=\frac{1}{9} u_{5x}-\frac{5}{9} u_{x} u_{2x}-\frac{5}{9} uu_{3x}
+\frac{5}{9} u^2 u_{x},
\\
\frac{\partial u}{\partial t_{11}}=
-\frac{1}{243} u_{11x}+\frac{11}{243} uu_{9x}-\frac{44}{243} u^2u_{7x}
+\frac{242}{729} u^3u_{5x}-\frac{220}{729} u^4u_{3x}+\frac{88}{729} u^5 u_{x}
\\ \hphantom{\frac{\partial u}{\partial t_{11}}=}
{}-\frac{110}{81} u^2 u_{x}^3+\frac{44}{243} u_{x}u_{8x}+\frac{121}{243} u_{2x}u_{7x}
+\frac{220}{243} u_{3x}u_{6x}-\frac{286}{243} uu_{x}u_{6x}
+\frac{286}{243} u_{4x}u_{5x}
\\ \hphantom{\frac{\partial u}{\partial t_{11}}=}
{}-\frac{616}{243} uu_{2x}u_{5x}-\frac{44}{27} u_{x}^2u_{5x}
-\frac{880}{243} uu_{3x}u_{4x}+\frac{616}{243} u^2u_{x}u_{4x}
+\frac{110}{27} u^2u_{2x}u_{3x}
\\ \hphantom{\frac{\partial u}{\partial t_{11}}=}
{}-\frac{440}{81} u_{2x}^2u_{3x}
+\frac{1298}{243} uu_{x}^2u_{3x}-\frac{979}{243} u_{x}u_{3x}^2
-\frac{1540}{729} u^3 u_{x} u_{2x}+\frac{572}{81} u u_{x} u_{2x}^2
\\ \hphantom{\frac{\partial u}{\partial t_{11}}=}
{}+\frac{682}{243} u_{x}^3 u_{2x}-\frac{1562}{243} u_{x}u_{2x}u_{4x}.
\end{gather*}
The equation corresponding to the flow $\frac{\partial}{\partial t_5}$ is the Sawada--Kotera equation~\cite{SK74}.

\subsubsection{Matrix resolvent and residues of fractional powers of Lax operators}\label{sec:LaxSK}
Let us consider the space $\mc R\big(\big(\lambda^{-1}\big)\big)^3$ and let us denote $\psi_1=\left(\begin{smallmatrix}0\\0\\1\end{smallmatrix}\right)$
and $\psi_2=\left(\begin{smallmatrix}0\\1\\0\end{smallmatrix}\right)$.
We have the following decomposition
\[
\mc R\big(\big(\lambda^{-1}\big)\big)^3=W_1\oplus W_2,
\]
where
\[
W_1=\widehat{\bigoplus}_{k\in\mb Z}\big(\mc R\otimes \Lambda(\lambda)^k\psi_1\big)
\qquad\text{and}\qquad
W_2=\widehat{\bigoplus}_{k\in\mb Z}\big(\mc R\otimes \Lambda(\lambda)^k\psi_2\big).
\]
This decomposition can be checked directly using the formulas
\[
\Lambda(\lambda)^{3k}=\lambda^{2k}\id_3,\qquad
\Lambda(\lambda)^{3k+1}=\lambda^{2k}\Lambda(\lambda),\qquad
\Lambda(\lambda)^{3k+2}=\lambda^{2k}\Lambda(\lambda)^2,\qquad
k\in\mb Z.
\]
Following \cite{DS85} we introduce a $\mc R((\partial^{-1}))$-module structure on $\mc R\big(\big(\lambda^{-1}\big)\big)^3$ by setting
\[
\partial^n.\eta(\lambda) := (\partial+q^{\rm can}+\Lambda(\lambda))^n(\eta(\lambda)),\qquad
\eta(\lambda)\in \mathcal{R}\big(\big(\lambda^{-1}\big)\big)^3,\quad n\in\mb Z.
\]
Note that $\Lambda(\lambda)$ is invertible,
hence the action of $\partial^{-1}$ is well defined using the geometric series expansion
\begin{align}
(\partial+q^{\rm can}+\Lambda(\lambda))^{-1} &=\big(\Lambda(\lambda)(\id_3+\Lambda(\lambda)^{-1}(\partial+q^{\rm can}))\big)^{-1}\nn
\\
&=\sum_{k\in\mb Z_{\geq0}}(-1)^k\big(\Lambda(\lambda)^{-1}(\partial+q^{\rm can})\big)^k\Lambda(\lambda)^{-1},\label{eq:20210708:eq1}
\end{align}
which gives a well-defined operator on $\mc R\big(\big(\lambda^{-1}\big)\big)^3$. Indeed, we note that multiplication by
\[
\Lambda(\lambda)^{-1}=
\begin{pmatrix} 0 & \lambda^{-1} & 0 \\ 0 & 0 &\lambda^{-1} \\
1 & 0 & 0\\
\end{pmatrix}
\]
does not
increase the orders of powers of~$\lambda$ of elements in $\R\big(\big(\lambda^{-1}\big)\big)^3$.

Since $q^{\rm can}(W_i)\subset W_i$ we have that $(\partial+q^{\rm can}+\Lambda(\lambda))(W_i)\subset W_i$, $i=1,2$.
Using the arguments in~\cite{DS85} we can show that any vector in $W_i$, $i=1,2$,
which contains only non-negative powers of~$\lambda$ can be uniquely expressed
as $A(\partial).\psi_i$, where $A(\partial)\in\mc R[\partial]$ is a differential operator.

Note that $\lambda\psi_1\in W_2$ and $\lambda\psi_2\in W_1$. Hence, there exist unique $L_1(\partial),L_2(\partial)\in\mc R[\partial]$ such that
\[
L_1(\partial).\psi_1=\lambda\psi_2
,
\qquad
L_2(\partial).\psi_2=\lambda\psi_1
.
\]
It is immediate to check that $L_1(\partial)=\partial^2-u$ and $L_2(\partial)=\partial$. Moreover, we have that $L_1(\partial)L_2(\partial)\psi_2=\lambda^2\psi_2$,
from which it follows (denoting $L(\partial)=L_1(\partial)L_2(\partial)$ and using the standard arguments in~\cite{DS85})
that $\big(\lambda^{4k}R_1(\lambda)\big)\psi_2=L(\partial)^{\frac{6k+1}{3}}.\psi_2$ and $\big(\lambda^{4k}R_5(\lambda)\big)\psi_2=L(\partial)^{\frac{6k+5}{3}}.\psi_2$, $k\in\mb Z_{\ge0}$.

Recall from \cite{DS85} that the DS hierarchy for $A_2^{(2)}$ and the $0$-th vertex of its Dynkin diagram
can be rewritten as
\[
\frac{\p L(\partial)}{\p t_i}= \big[L(\partial),\big(L(\partial)^{\frac{i}{3}}\big)_+\big],
\]
where $i\in E\cap \mathbb{Z}_{>0}$ (recall that $i=1+6k$ or $i=5+6k$, $k\geq0$).

\begin{Proposition}\label{prop:SK}
For $a=1,2$ and $k\in\mb Z_{\ge0}$, we have
\[
\Omega_{a,k;1,0}=2\Res_\partial L(\partial)^{\frac{m_a+6k}{3}},
\]
where $\Res_\p L(\partial)^{\frac{m_a+6k}{3}}$ denotes the coefficient of $\partial^{-1}$ of the pseudodifferential oeprator $L(\partial)^{\frac{m_a+6k}{3}}$.
\end{Proposition}
\begin{proof}
Note that, if $A(\partial)\in\R[\partial]$, then $A(\partial).\psi_2\in\mc R[\lambda]^3$. Hence,
using \eqref{eq:20210708:eq1} we have that
\begin{equation}\label{eq:20210728:eq2}
L(\partial)^{\frac{6k+m_a}{3}}.\psi_2=\text{a polynomial in $\lambda$}
+\begin{pmatrix}
\Res_\partial L(\partial)^{\frac{6k+m_a}{3}}
\\
*
\\
*
\end{pmatrix}\lambda^{-1}+O\big(\lambda^{-2}\big).
\end{equation}
On the other hand, let $R_{m_a}(\lambda)$ be written as in~\eqref{eq:R-KK}. Then we have that
\begin{equation}\label{eq:20210728:eq3}
\big(\lambda^{4k}R_{m_a}(\lambda)\big)\psi_2=\begin{pmatrix}b_1(\lambda)+b_2(\lambda)
\\-2a_2(\lambda)\\c_1(\lambda)-c_2(\lambda)\end{pmatrix}\lambda^{4k}.
\end{equation}
Recall that $b_2(\lambda)\in\R((\lambda^{-4}))\lambda$, hence $\Res_{\lambda}b_2(\lambda)=0$.
Since $\left(\lambda^{4k}R_{m_a}(\lambda)\right)\psi_2=L(\partial)^{\frac{6k+m_a}{3}}.\psi_2$, from equations \eqref{eq:20210728:eq2} and \eqref{eq:20210728:eq3} we then have
\[
2\Res_\p L(\partial)^{\frac{6k+m_a}{3}}=\Res_\lambda 2b_1(\lambda)\lambda^{4k}
=\Res_\lambda \left(R_{m_a}(\lambda)|\tilde e_0\right)\lambda^{4k}
=\Res_{\lambda}G_a(\lambda)\lambda^{4k}=\Omega_{a,k;1,0}.
\]
In the last identity we used \eqref{eq:G}. This concludes the proof.
\end{proof}

\begin{Remark}\label{rem:SK}
It is claimed in \cite{DS85} that the differential polynomials
$\Res_{\partial}L^{\frac{m_a+6k}{3}}(\partial)$, $k\geq0$, $a=1,2$, up to constant multiples can be served as Hamiltonian
densities for the Hamiltonian structure of the SK hierarchy (for a review of the Hamiltonian formalism of DS hierarchies~\cite{DS85} see Section~\ref{sec:Ham}).
\end{Remark}

\subsection{The Kaup--Kupershmidt hierarchy}\label{sec:KK}
The Drinfeld--Sokolov hierarchy associated to $A_2^{(2)}$ and the vertex $c_1$ of its Dynkin diagram~\eqref{diagram} is known as the Kaup--Kupershmidt hierarchy~\cite{Kaup80}. Following Section~\ref{sec:matrix_res} we compute the basic matrix resolvents for this hierarchy.

\subsubsection[Principal and standard gradations for A\_2\textasciicircum{}\{(2)\} and the c\_1 vertex]{Principal and standard gradations for $\boldsymbol{A_2^{(2)}}$ and the $\boldsymbol{c_1}$ vertex}

In this case there exists an automorphism $\sigma_1$ of $\mf{sl}_3$ of order $N_1=2$ (cf.~Section~\ref{sec:realizationKM})
such that $\widetilde{\mf g}=L(\mf{sl}_3,\sigma_1)\subset L(\mf{sl}_3)$:
\begin{gather*}
\widetilde{\mf g}=\left\{\!\!
\left.
\begin{pmatrix}
a_1(\lambda)+a_2(\lambda)&b_1(\lambda)+b_2(\lambda)&p(\lambda)\\
c_1(\lambda)+c_2(\lambda)&-2a_2(\lambda)&b_1(\lambda)-b_2(\lambda)\\
r(\lambda)&c_1(\lambda)-c_2(\lambda)&a_2(\lambda)-a_1(\lambda)
\end{pmatrix}
\right|
\begin{array}{l}
a_1(\lambda),b_1(\lambda),c_1(\lambda)\in\mb C((\lambda^{-2}))
\\
a_2(\lambda),b_2(\lambda),c_2(\lambda),
\\ \quad
p(\lambda),r(\lambda)\in\mb C((\lambda^{-2}))\lambda
\end{array}\!\!
\right\}\!.
\end{gather*}
Let us consider the following Chevalley generators for $\widetilde{\mf g}$:
\begin{gather}
e_0=E_{12}+E_{23},\quad\,
h_0=2(E_{11}-E_{33}),\quad
f_0=2(E_{21}+E_{32}),\nn
\\
e_1(\lambda)=E_{31}\lambda,\qquad
h_1=E_{33}-E_{11},\qquad\,
f_1(\lambda)=E_{13}\lambda^{-1} .\label{gen_c1}
\end{gather}
The principal gradation is defined by the linear map \eqref{deg:principal3}, where $\rho^\vee=h_0/2$. Explicitly, we have
\begin{align*}
&\widetilde{\mf g}^{6k}=\mb C h_0\lambda^{2k},\qquad
\widetilde{\mf g}^{6k+1}=\mb C e_0\lambda^{2k}\oplus\mb C e_1(\lambda)\lambda^{2k},\qquad
\widetilde{\mf g}^{6k+2}=\mb C (E_{21}-E_{32})\lambda^{2k+1},
\\
&\widetilde{\mf g}^{6k+3}=\mb C (E_{11}-2E_{22}+E_{33})\lambda^{2k+1},\qquad
\widetilde{\mf g}^{6k+4}=\mb C (E_{12}-E_{23})\lambda^{2k+1},
\\
& \widetilde{\mf g}^{6k+5}=\mb C f_0\lambda^{2k+2}\oplus\mb C f_1(\lambda)\lambda^{2k+2},
\end{align*}
where $k\in\mb Z$.
The standard gradation corresponding to the vertex $c_1$ is the gradation in powers of $\lambda$. Hence we have
\begin{gather*}
\widetilde{\mf g}_{2k}=\mb C f_0\lambda^{2k}\oplus\mb C h_0\lambda^{2k}\oplus\mb C e_0\lambda^{4k},
\\
\widetilde{\mf g}_{2k+1}=\mb C f_1\lambda^{2k}\oplus\mb C(E_{21}-E_{32})\lambda^{2k+1}\oplus
\mb C (E_{11}-2E_{22}+E_{33})\lambda^{2k+1}
\\ \hphantom{\widetilde{\mf g}_{2k+1}=}
{}\oplus\mb C (E_{12}-E_{23})\lambda^{2k+1}\oplus\mb C e_1\lambda^{2k}.
\end{gather*}
Note that $\mf a=\widetilde{\mf g}_0\cong\mf{sl}_2$,
moreover, $\mf n = \mb Cf_0 \subset \mf b= \mb C f_0 \oplus \mb C h_0$.

Recall that the element $\Lambda(\lambda)=e_0+e_1(\lambda)\in\widetilde{\mf g}$ is semisimple and
we have the direct sum decomposition~\eqref{dec:Lambda}. Let, as in Section~\ref{section1},
$\mc H=\Ker \ad\Lambda(\lambda)$. It is immediate to check that
\begin{equation}\label{eq:Lambda-c1}
\mc H=\Bigl(\widehat\bigoplus_{k\in\mb Z}\mb C\Lambda(\lambda)\lambda^{2k}\Bigr)\oplus
\Bigl(\widehat\bigoplus_{k\in\mb Z}\mb C(f_0+2f_1(\lambda))\lambda^{2k}\Bigr).
\end{equation}
We rewrite~\eqref{eq:Lambda-c1} as in~\eqref{mcHdecompbasislambda}
using the following basis $\{\Lambda_i(\lambda)\mid i\in E\}$:
\begin{gather*}
\Lambda_1(\lambda)=\Lambda(\lambda),
\ \ \Lambda_5(\lambda)=\bigg(\!\frac12f_0+f_1(\lambda)\!\bigg)\lambda^2,
\ \ \Lambda_{1+6k}(\lambda)=\Lambda_1(\lambda)\lambda^{2k},
\ \ \Lambda_{5+6k}(\lambda)=\Lambda_5(\lambda)\lambda^{2k},
\end{gather*}
where $k\in\mb Z$.
This basis satisfies the normalization conditions \eqref{eq:normalization1} and \eqref{eq:normalization}.

\subsubsection{The matrix resolvent}
Take the DS gauge $V=\mb Cf_0$ (cf.~\eqref{bVen}). The element $\mc L^{\text{can}}$ in~\eqref{Lcan} has the form
\[
\mc L^{\text{can}}=\partial+\Lambda(\lambda)+u \frac{f_0}2
=\begin{pmatrix}
\partial&1&0\\
u&\partial&1\\
\lambda&u&\partial
\end{pmatrix}\!.
\]
We have that $\R=\mb C[u,u_x,u_{2x},\dots]$ is the algebra of differential polynomials
in~$u$.
Let
\begin{equation}\label{eq:R-KK}
R(\lambda)=\begin{pmatrix}
a_1(\lambda)+a_2(\lambda)&b_1(\lambda)+b_2(\lambda)&p(\lambda)\\
c_1(\lambda)+c_2(\lambda)&-2a_2(\lambda)&b_1(\lambda)-b_2(\lambda)\\
r(\lambda)&c_1(\lambda)-c_2(\lambda)&a_2(\lambda)-a_1(\lambda)
\end{pmatrix}
\in \mc R\otimes \widetilde{\mf g}
\end{equation}
be a resolvent of~$\mc L^{\text{can}}$.
Then $[\mc L^{\text{can}},R(\lambda)]=0$; cf.~\eqref{LRcomm}.
Solving this linear system we find that
\begin{align*}
&a_1(\lambda)=b_{1,x}(\lambda),\qquad a_2(\lambda)=-\frac13 up(\lambda)+\frac16p_{2x}(\lambda),\qquad b_2(\lambda)=\frac12p_x(\lambda) ,
\\
&c_1(\lambda)=ub_1(\lambda)+p(\lambda)\lambda -b_{1,2x}(\lambda)
,
\qquad
c_2(\lambda)=\frac{1}{3}u_x p(\lambda)+\frac56 up_x(\lambda)-\frac{1}{6}p_{3x}(\lambda)
,
\\
&r(\lambda)=b_1(\lambda)\lambda+\frac13\big(3u^2-u_{2x}\big)p(\lambda)-\frac76 u_x p_x(\lambda)-\frac43 u p_{2x}(\lambda)+\frac16p_{4x}(\lambda)
,
\end{align*}
where
$b_1=b_1(\lambda)\in\mc R((\lambda^{-2}))\lambda$ and $p=p(\lambda)\in\mc R((\lambda^{-2}))$ satisfy the following system of ODEs:
\begin{equation}\label{eq:system1bis}
\begin{cases}
18 b_{1,x}\lambda+2(8u u_x- u_{3x})p +\big(16 u^2 -9 u_{2x}\big)p_x -15 u_x p_{xx}-10 up_{3x}+p_{5x}=0,
\\
3p_x\lambda+2u_xb_1+4ub_{1,x}-2b_{1,3x}=0.
\end{cases}
\end{equation}
To find the basic resolvent $R_1(\lambda)$, we write it as in~\eqref{eq:R_a-principal} as follows:
\[
R_1(\lambda)=\Lambda(\lambda)+\text{ terms of lower degree}.
\]
Then one can solve \eqref{eq:system1bis} for $p(\lambda)=\sum_{k\geq0}p_{k}\lambda^{-2k-1}$ and $b_1(\lambda)=\sum_{k\geq0}b_{1,k}\lambda^{-2k}$
recursively. The first few terms are as follows:
\begin{gather}
p(\lambda)=-\frac23u\lambda^{-1}+\frac{1}{243}\big(6 u_{6x}-84 uu_{4x}+336 u^2 u_{2x}-147 u_{2x}^2
+420 u u_{x}^2\nn
\\[1mm] \hphantom{p(\lambda)=}
{}-210 u_{x}u_{3x}-112 u^4\big)\lambda^{-3}+\cdots,\notag
\\[1mm] b_1(\lambda)=1+\frac1{81}\big(32u^3-18u_{x}^2-36uu_{2x}+3u_{4x}\big)\lambda^{-2}
+\frac1{6561}\big({-}9 u_{10x}+216 uu_{8x}\notag
\\[1mm] \hphantom{b_1(\lambda)=}
{}-1908 u^2u_{6x}+7728 u^3u_{4x}+2718 u_{4x}^2-15174 u u_{3x}^2-15120 u^4 u_{2x}+34776 u^2 u_{2x}^2\notag
\\[1mm] \hphantom{b_1(\lambda)=}
{}-11463 u_{2x}^3-30240 u^3 u_{x}^2+7749 u_{x}^4+864 u_{x}u_{7x}+2493 u_{2x}u_{6x}
+4455 u_{3x}u_{5x}\notag
\\[1mm] \hphantom{b_1(\lambda)=}
{}-11448 uu_{x}u_{5x}-24714 uu_{2x}u_{4x}-14067 u_{x}^2u_{4x}
+46368 u^2u_{x}u_{3x}\nn
\\[1mm] \hphantom{b_1(\lambda)=}
{}+77364 u u_{x}^2 u_{2x}-48456 u_{x}u_{2x}u_{3x}+2240 u^6\big)\lambda^{-4}+\cdots.\label{p-KK}
\end{gather}
Similarly, we write
\[
R_5(\lambda)=\Lambda_5(\lambda)+\text{terms of lower degree},
\]
and we solve \eqref{eq:system1bis} for $p(\lambda)=\lambda +\sum_{k\geq0}p_{k}\lambda^{-2k-1}$
and $b_1(\lambda)=\sum_{k\geq0}b_{1,k}\lambda^{-2k}$
recursively. The first few terms are as follows
\begin{gather}
p(\lambda)=\lambda+\frac1{81}\big(40u^3-45u_{x}^2-60uu_{2x}+6u_{4x}\big)\lambda^{-1}+
\frac{1}{6561}\big(150480 u u_{x}^2 u_{2x}\notag
\\[1mm] \hphantom{p(\lambda)=}
{}-100188 u_{x} u_{2x}u_{3x}-47520 u^3 u_{x}^2+77616 u^2 u_{x}u_{3x}-21384 u u_{x}u_{5x}
+19602 u_{x}^4\notag
\\[1mm] \hphantom{p(\lambda)=}
{}-30888 u_{x}^2u_{4x}+1782 u_{x}u_{7x}-21120 u^4 u_{2x}+56232 u^2 u_{2x}^2-44352 u u_{2x}u_{4x}\notag
\\[1mm] \hphantom{p(\lambda)=}
{}-22044 u_{2x}^3+4950 u_{2x}u_{6x}+11616 u^3u_{4x}
-3168 u^2u_{6x}-27324 u u_{3x}^2+396 u u_{8x}\nn
\\[1mm] \hphantom{p(\lambda)=}
{}+5445 u_{4x}^2+8910 u_{3x} u_{5x}-18 u_{10x}+2816 u^6\big)\lambda^{-3}+\cdots,\notag
\\[1mm]
b_1(\lambda)=\frac1{9}\big(u_{2x}-4u^2\big)
+\frac{1}{729}\big({-}3 u_{8x}+60 uu_{6x}-408 u^2u_{4x}+252 u_{3x}^2
+1120 u^3 u_{2x}\notag
\\ \hphantom{b_1(\lambda)=}
{}-1224 u u_{2x}^2+1680 u^2 u_{x}^2+180 u_{x}u_{5x}+402 u_{2x}u_{4x}
-1632 uu_{x}u_{3x}\notag
\\ \hphantom{b_1(\lambda)=}
{}-1188 u_{x}^2 u_{2x}-256 u^5\big)\lambda^{-2}+\cdots.\label{p-KK2}
\end{gather}
Recall from Section~\ref{sec:3.2} that $\Lambda(\lambda)=e+e_1(\lambda)$, where $e=e_0$ and
$e_1(\lambda)=\tilde e_1\lambda$. From \eqref{gen_c1} we
have that $\tilde e_1=E_{31}$. Let $R(\lambda)$ be as in \eqref{eq:R-KK}, then $(R(\lambda)|\tilde e_1)=p(\lambda)$.
Hence, recalling the definition of the series $G_a(\lambda)$, $a=1,2$, given in \eqref{eq:G} and
\eqref{p-KK}--\eqref{p-KK2} we have the following expression for the first few terms of the
tau-structure of the DS hierarchy
\begin{gather}
\Omega_{1,0;1,0}=-\frac{2}{3}u,\nn
\\
\Omega_{1,1;1,0}=\frac{1}{243}\big(6 u_{6x}-84 uu_{4x}+336 u^2 u_{2x}-147 u_{2x}^2
+420 u u_{x}^2-210 u_{x}u_{3x}-112 u^4\big),\nn
\\
\Omega_{2,0;1,0}=\frac1{81}\big(40u^3-45u_{x}^2-60uu_{2x}+6u_{4x}\big),\nn
\\
\Omega_{2,1;1,0}=\frac{1}{6561}
\big(150480 u u_{x}^2 u_{2x}-100188 u_{x} u_{2x}u_{3x}
-47520 u^3 u_{x}^2+77616 u^2 u_{x}u_{3x}\notag
\\ \hphantom{\Omega_{2,1;1,0}=}
{}-21384 u u_{x}u_{5x}+19602 u_{x}^4-30888 u_{x}^2u_{4x}
+1782 u_{x}u_{7x}-21120 u^4 u_{2x}\notag
\\ \hphantom{\Omega_{2,1;1,0}=}
{}+56232 u^2 u_{2x}^2-44352 u u_{2x}u_{4x}-22044 u_{2x}^3+4950 u_{2x}u_{6x}+11616 u^3u_{4x}\nn
\\ \hphantom{\Omega_{2,1;1,0}=}
{}-3168 u^2u_{6x}-27324 u_{3x}^2 u+396 u u_{8x}+5445 u_{4x}^2+8910 u_{3x} u_{5x}\nn
\\ \hphantom{\Omega_{2,1;1,0}=}
{}-18 u_{10x}+2816 u^6\big).\label{tau:KK}
\end{gather}
Similarly to what done in Section~\ref{sec:SKhierarchy}, one can show that the DS hierarchy for $A_2^{(2)}$ and $c_1$ can be written again using equation \eqref{latestchange}.
Hence, from \eqref{latestchange} and \eqref{tau:KK} we get the first few equations of the hierarchy:
\begin{gather*}
\frac{\partial u}{\partial t_1}=-u_{x},
\\
\frac{\partial u}{\partial t_6}=
\frac{1}{27} u_{7x}-\frac{14}{27} uu_{5x}+\frac{56}{27} u^2u_{3x}
-\frac{224}{81} u^3 u_{x}+\frac{70}{27} u_{x}^3\!-\frac{49}{27} u_{x}u_{4x}\!
-\frac{28}{9} u_{2x}u_{3x}\!+\frac{28}{3} u u_{x}u_{2x},
\\
\frac{\partial u}{\partial t_5}=\frac{1}{9} u_{5x}-\frac{25}{9} u_{x} u_{2x}-\frac{10}{9} uu_{3x}
+\frac{20}{9} u^2 u_{x},
\\
\frac{\partial u}{\partial t_{11}}=
-\frac{1}{243} u_{11x}+\frac{22}{243} uu_{9x}-\frac{176}{243} u^2u_{7x}
+\frac{1936}{729} u^3u_{5x}-\frac{3520}{729} u^4u_{3x}+\frac{2816}{729} u^5 u_{x}
\\ \hphantom{\frac{\partial u}{\partial t_5}=}
{}-\frac{880}{27} u^2 u_{x}^3+\frac{121}{243} u_{x}u_{8x}+\frac{374}{243} u_{2x}u_{7x}
+\frac{770}{243} u_{3x}u_{6x}-\frac{1540}{243} uu_{x}u_{6x}
+\frac{1100}{243} u_{4x}u_{5x}
\\ \hphantom{\frac{\partial u}{\partial t_5}=}
{}-\frac{3652}{243} uu_{2x}u_{5x}-\frac{968}{81} u_{x}^2u_{5x}
{}-\frac{5500}{243} uu_{3x}u_{4x}+\frac{6248}{243} u^2u_{x}u_{4x}
+\frac{3520}{81} u^2u_{2x}u_{3x}
\\ \hphantom{\frac{\partial u}{\partial t_5}=}
{}-\frac{3080}{81} u_{2x}^2u_{3x}
+\frac{16984}{243} u u_{x}^2u_{3x}-\frac{7084}{243} u_{x} u_{3x}^2
-\frac{29920}{729} u^3 u_{x} u_{2x}+\frac{2552}{27} u u_{x} u_{2x}^2
\\ \hphantom{\frac{\partial u}{\partial t_5}=}
{}+\frac{12716}{243}u_{x}^3 u_{2x}
-\frac{11462}{243} u_{x}u_{2x}u_{4x}.
\end{gather*}
The equation corresponding to the flow $\frac{\partial}{\partial t_5}$ is the Kaup--Kupershmidt equation~\cite{Kaup80}.

\subsubsection{Matrix resolvent and residues of fractional powers of Lax operators} \label{subsubsectionkk}
Let us consider the space $\mc R\big(\big(\lambda^{-1}\big)\big)^3$.
Following \cite{DS85} we introduce a $\mc R\big(\big(\partial^{-1}\big)\big)$-module structure on $\mc R\big(\big(\lambda^{-1}\big)\big)^3$ by setting
\[
\partial^n.\eta(\lambda)=\big(\partial+q^{\rm can}+\Lambda(\lambda)\big)^n(\eta(\lambda)),\qquad
\eta(\lambda)\in\mc R\big(\big(z^{-1}\big)\big)^3,\quad
n\in\mb Z.
\]
Note that $\Lambda$ is invertible, hence the action of $\partial^{-1}$ is well defined using the geometric series expansion
as in Section~\ref{sec:LaxSK} (cf.~\eqref{eq:20210708:eq1}).

Let $\psi=\left(\begin{smallmatrix}0\\0\\1\end{smallmatrix}\right)$.
Using the arguments in \cite{DS85} we can show that any vector in $\mc R[\lambda]^3$
can be uniquely expressed
as $A(\partial).\psi$, where $A(\partial)\in\mc R[\partial]$ is a differential operator.
Hence, there exists unique $L(\partial)\in\mc R[\partial]$ such that
\[
L(\partial).\psi=\lambda\psi
.
\]
It is straightforward to check that $L(\partial)=\partial^3-u\partial-\partial\circ u$
from which it follows, using the standard arguments in \cite{DS85},
that $\big(\lambda^{2k}R_1(\lambda)\big)\psi=L(\partial)^{\frac{6k+1}{3}}.\psi$ and $\big(\lambda^{2k}R_5(\lambda)\big)\psi=L(\partial)^{\frac{6k+5}{3}}.\psi$, $k\in\mb Z_{\geq0}$.

Recall from \cite{DS85} that the DS hierarchy for $A_2^{(2)}$ and vertex $c_1$ of its Dynkin diagram
can be rewritten as
\[
\frac{\partial L(\partial)}{\partial t_i}=\big[L(\partial),\big(L(\partial)^{\frac{i}{3}}\big)_+\big],
\]
where $i\in E\cap\mb Z_{\geq1}$.
\begin{Proposition}\label{prop:KK}
For $a=1,2$ and $k\in\mb Z_+$, we have
\[
\Omega_{a,k;1,0}=\Res_\partial L(\partial)^{\frac{m_a+6k}{3}}.
\]
\end{Proposition}
\begin{proof}
The argument is the same as in the proof of Proposition \ref{prop:SK}.
Recall that, if $A(\partial)\in\mc R[\partial]$, then $A(\partial).\psi\in\mc R[\lambda]^3$. Hence,
using \eqref{eq:20210708:eq1} we have that
\begin{equation}\label{eq:20210728:eq2b}
L(\partial)^{\frac{6k+m_a}{3}}.\psi=\text{a polynomial in $\lambda$}
+\begin{pmatrix}
\Res_\partial L(\partial)^{\frac{6k+m_a}{3}}
\\
0
\\
*
\end{pmatrix}\lambda^{-1}+O\big(\lambda^{-2}\big).
\end{equation}
On the other hand, let $R_{m_a}(\lambda)$ be written as in \eqref{eq:R-KK}. Then we have that
\begin{equation}\label{eq:20210728:eq3b}
(\lambda^{2k}R_{m_a}(\lambda))\psi=\begin{pmatrix}p(\lambda)\\b_1(\lambda)-b_2(\lambda)
\\a_2(\lambda)-a_1(\lambda)\end{pmatrix}\lambda^{2k}.
\end{equation}
Since $\big(\lambda^{2k}R_{m_a}(\lambda)\big)\psi=L(\partial)^{\frac{6k+m_a}{3}}.\psi$, from equations \eqref{eq:20210728:eq2b} and \eqref{eq:20210728:eq3b} we then have
\[
\Res_\partial L^{\frac{6k+m_a}{3}}=\Res_\lambda p(\lambda)\lambda^{2k}=\Res_\lambda(R_{m_a}(\lambda)|\tilde e_1)\lambda^{2k}
=\Res_{\lambda}G_a(\lambda)\lambda^{4k}=\Omega_{a,k;1,0}
.
\]
In the last identity we used \eqref{eq:G}. This concludes the proof.
\end{proof}

\begin{Remark}
Similarly to Remark~\ref{rem:SK} we have that the differential polynomials
$\Res_{\partial}L^{\frac{m_a+6k}{3}}(\partial)$, $k\geq0$, $a=1,2$, up to constant multiples can be served as Hamiltonian
densities for the Hamiltonian structure of the KK hierarchy.
\end{Remark}

\section[Hamiltonian structure and tau-structure]{Hamiltonian structure and tau-structure\footnote{Section~\ref{sec:Ham} is written by D.V. and D.Y. after Boris Dubrovin passed away.
This section contains a generalization of Propositions~\ref{prop:SK} and~\ref{prop:KK} for the
explicit $A_2^{(2)}$ examples, and a
proof of Theorem~\ref{thmunt} that generalizes the $A_1^{(1)}$ case. }
} \label{sec:Ham}

\subsection{Review of known results on Hamiltonian structures of the DS hierarchies}

\subsubsection{Poisson structures}\label{sec:km}
Recall from~\cite{DS85} that (up to a constant factor) there exists a unique non-degenerate symmetric invariant $\mb C$-valued
bilinear form $\kappa\colon\widetilde{\mf g}\times \widetilde{\mf g}\to\mb C$ which is coordinated with the principal and standard gradations, that
is, $\kappa(a,b)=0$ if $a\in\widetilde{\mf g}^k$, $b\in\widetilde{\mf g}^l$ and $k+l\neq0$ (similarly for the standard gradation).
We note that the direct sum decomposition~\eqref{dec:Lambda} is orthogonal with respect to~$\kappa$.
The restriction
$(\cdot\,|\,\cdot)=\kappa|_{\mf a\times\mf a}$ of $\kappa$
to the semisimple subalgebra $\mf a$ is a non-degenerate symmetric invariant bilinear form~\cite{DS85}.
Let us extend the bilinear form $(\cdot\,|\,\cdot)$ on~$\mf a$
 to a bilinear form
(which we still denote with the same symbol with an abuse of notation) on smooth functions
$u=u(x)$, $v=v(x)\in C^{\infty}(S^1,\mf a)$, from the circle $S^1$ in $\mf a$, in the natural way by letting
\begin{equation}\label{uvequaluv}
(u|v)=\int (u(x)|v(x))\,{\rm d}x.
\end{equation}
Here and below, for a smooth function $g(x)$, we denote
\[
\int g(x)\,{\rm d}x:=\frac{1}{2\pi}\int_{S^1} g(x)\,{\rm d}x.
\]
Recall from Section~\ref{section1} that $\mc R\subset\mc A^q$ denotes the space of gauge invariants. Let
\[
\mc F=\bigl\{\bar f[q]:=\tint f(q,q_x,q_{2x},\dots)\,{\rm d}x\mid f\in\mc R\bigr\}
\]
be the space of local functionals whose densities are gauge invariant differential polynomials.
The space~$\mc F$ can be canonically identified
with the quotient space $\mc R/\partial\mc R$.

Let $\bar f := \tint f(q,q_x,q_{2x},\dots)\,{\rm d}x\in\mc F$. Then, $f(q,q_x,q_{2x},\dots)\in\mc R \subset\mc A^q$ is a differential polynomial
in the entries of~$q$, and their $x$-derivatives.

Let $V$ be a Drinfeld--Sokolov gauge (cf.~\eqref{bVen}).
Let $v_1,\dots,v_{\dim \mf b}$ be a basis of $\mf b$ homogeneous with respect to~\eqref{dec:dynkin},
such that $v_1, \dots, v_\ell$ is a basis for~$V$, and $v_{\ell+1},\dots,v_{\dim \mf b}$
is a basis for~$[e,\mf n]$.
Write $q=\sum_{i=1}^{\dim \mf b}q_iv_i$.
We define the variational derivative of~$\bar f$ with respect to $q_i$ as
\begin{equation}\label{varder}
\frac{\delta \bar f}{\delta q_i}=\sum_{k\geq0}(-\partial)^k\frac{\partial f}{\partial (\partial^kq_i)},
\qquad i=1,\dots,\dim \mf b .
\end{equation}
Let $v^1,\dots,v^{\dim \mf b}$ be the basis of~$\mf a^{\geq0}$, which is dual, with respect to the non-degenerate bilinear form $(\cdot\,|\,\cdot)$
on $\mf a$, to the basis
$v_1,\dots,v_{\dim \mf b}$ of $\mf b$. Recall from \cite[Lemma~3.13]{DSKV13} that $(v^i)_{i=1,\dots,\ell}$ is a basis of
$\mf a^e=\{a\in \mf a \mid [a,e]=0\}$, the centralizer
of~$e$ in $\mf a$.
We identify the collection $\frac{\delta \bar f}{\delta q}=\big(\frac{\delta \bar f}{\delta q_i}\big)_{i=1,\dots,\dim \mf b}$
of all variational derivatives \eqref{varder} of~$\bar f$ with respect to the variables~$q_i$ with a smooth functions with values in $\mf a^{\geq0}$
by
\begin{equation}\label{eq:varder}
\frac{\delta \bar f}{\delta q}=\sum_{i=1}^{\dim \mf b} \frac{\delta \bar f}{\delta q_i}v^i.
\end{equation}
In \cite{DS85} it is shown that the following formula:
\begin{align}\label{eq:Poisson1}
\bigl\{\bar f, \bar g\bigr\}_2[q]=\int\biggl(\frac{\delta \bar f}{\delta q}\bigg|\biggl[\frac{\delta \bar g}{\delta q},\partial+f+q\biggr]\biggr)\,{\rm d}x,
\qquad \forall\, \bar f,\bar g\in\mc F,
\end{align}
defines a local Poisson bracket on $\mc F$ (namely a Lie algebra structure on $\mc F$).
In particular, for any gauge transformation~\eqref{gauge2} given by $N\in\mc A^q\otimes \mf n$,
$\big\{\bar f,\bar g\big\}_2\big[q^N\big]=\big\{\bar f,\bar g\big\}_2[q]$, so it is a~well-defined element of $\mc F$.

Recall from the Section~\ref{section1} that, $f(q,q_x,q_{xx},\dots)\in\mc R$ is a differential polynomial in the
variables $u_i$, $i=1,\dots,\ell$ \big($q^{\rm can}=\sum_{i=1}^\ell u_i v_i$, cf.~\eqref{Lcan}\big).
Hence, we can also consider the vector of variational derivatives
\[
\frac{\delta \bar f}{\delta q^{\rm can}}=\left(\frac{\delta \bar f}{\delta u_i}\right)_{i=1}^\ell,
\]
where $\frac{\delta \bar f}{\delta u_i}$ is defined as in~\eqref{varder}, where $q_i$ is replaced by~$u_i$, $i=1,\dots,\ell$.
As before, we will also use the following notation
\begin{equation}\label{eq:vardercan}
\frac{\delta \bar f}{\delta q^{\rm can}}=\sum_{i=1}^\ell \frac{\delta \bar f}{\delta u_i} v^i
\end{equation}
 to identify the variational derivative with a function with values in~$\mf a^e$.
The Lie bracket~\eqref{eq:Poisson1} can be rewritten as
\begin{align}\label{eq:Poisson1b}
\bigl\{\bar f,\bar g\bigr\}_2(q^{\rm can})=\int \frac{\delta \bar f}{\delta q^{\rm can}}P(\partial)\frac{\delta \bar g}{\delta q^{\rm can}}\,{\rm d}x,
\qquad \bar f,\bar g\in\mc F,
\end{align}
where $P(\partial)=(P_{ij}(\partial))_{i,j=1}^\ell$ is a local Hamiltonian operator~\cite{DSKV13, DSKV16,DS85, DZ-norm}.

Let us extend the bilinear form~$\kappa$ on $\widetilde{\mf g}$ to a bilinear form
on smooth functions with values in~$\widetilde{\mf g}$ in the natural way as in~\eqref{uvequaluv}.
For any $\Theta\in\mc H$, define $\bar h_\Theta =\tint \kappa(\Theta,H)\,{\rm d}x\in\mc F$ (we review the proof of the fact
$\bar h_\Theta $ belongs to $\mc F$ in Section~\ref{sec:Ham1}).
Here $H$ is defined in~\eqref{eq:L0}.
Then, the DS hierarchy~\eqref{DS:hier} can be written in Hamiltonian form (using \eqref{eq:Poisson1} or \eqref{eq:Poisson1b}) as
\begin{equation}\label{DShierarinham2}
\frac{\p u_s}{\p t_\Theta}= \bigl\{\bar h_\Theta,u_s(x)\bigr\}_2=\sum_{j=1}^\ell P_{s j}(\partial)\frac{\delta \bar h_\Theta}{\delta u_j},\qquad
s=1,\dots,\ell, \quad \Theta\in\mc H.
\end{equation}
It is proved in~\cite{DS85} that $\big\{\bar h_{\Theta_1},\bar h_{\Theta_2}\big\}_2=0$, for every $\Theta_1,\Theta_2\in\mc H$.
\begin{Example}
For both the SK hierarchy (see Section~\ref{sec:SKhierarchy}) and the KK hierarchy (see Section~\ref{sec:KK}), we
have that $\mc R=\mb C[u,u_x,u_{xx},\dots]$.
The Poisson structure \eqref{eq:Poisson1b} for the SK hierarchy is given by the Hamiltonian operator
$P(\partial)=-u_x-2u\partial+\frac{1}{2}\partial^3$,
while the
Poisson structure \eqref{eq:Poisson1b} for the KK hierarchy is given by the Hamiltonian operator
$P(\partial)=-\frac{u_x}2-u\partial+\frac12\partial^3$ (see~\cite{DSKV13,DS85}).
\end{Example}

Recall also from \cite{DS85}, that for untwisted affine Kac--Moody algebras (this is the case when $r=1$) and the choice of the special vertex of the
Dynkin diagram $c_0$, we have that $\widetilde{\mf g}=\mf a\big(\big(\lambda^{-1}\big)\big)$
and $e_0=e_{-\theta}\lambda$, $e_{-\theta}$ being the lowest root vector of~$\mf a$.
(With this realization, the standard gradation defined in Section~\ref{section1} coincides with the gradation in powers of~$\lambda$.)
In this case it is possible to endow~$\mc F$ with another
local Poisson bracket compatible with~\eqref{eq:Poisson1}. It is defined as
\begin{gather*}
\bigl\{\bar f,\bar g\bigr\}_1[q]=\int\bigg(\frac{\delta \bar f}{\delta q}\bigg|\bigg[e_{-\theta},\frac{\delta \bar g}{\delta q}\bigg]\bigg){\rm d}x,
\qquad \forall\, \bar f,\bar g\in\mc F.
\end{gather*}
As in the previous discussion, this can be rewritten as
\begin{align}\label{eq:Poisson2b}
\bigl\{\bar f,\bar g\bigr\}_1[q^{\rm can}]=\int \frac{\delta \bar f}{\delta q^{\rm can}}Q(\partial)\frac{\delta \bar g}{\delta q^{\rm can}}\,{\rm d}x,
\qquad \bar f,\bar g\in\mc F,
\end{align}
where $Q(\partial)=(Q_{i,j}(\partial))_{i,j=1}^\ell$ is another local Hamiltonian operator compatible to $P(\partial)$~\cite{DSKV13, DSKV16,DS85,DZ-norm}.

Then, the DS hierarchy \eqref{DS:hier} can be written in another Hamiltonian form as
\begin{equation}\label{firsthamDS}
\frac{\p u_s}{\p t_\Theta}= \bigl\{\bar h_{\lambda\Theta},u_s(x)\bigr\}_1=\sum_{j=1}^\ell Q_{sj}(\partial)\frac{\delta \bar h_{\lambda\Theta}}{\delta u_j},\qquad
s=1,\dots,\ell, \quad \Theta\in\mc H,
\end{equation}
and we have that $\big\{\bar h_{\Theta_1},\bar h_{\Theta_2}\big\}_1=0$, for every $\Theta_1,\Theta_2\in\mc H$.
Furthermore, the following Lenard--Magri recursion relation \cite{M78} holds:
\[
\bigl\{\bar h_{\Theta},u(x)\bigr\}_2= \bigl\{\bar h_{\lambda\Theta},u(x)\bigr\}_1,
\]
for every $u(x)\in\mc R$ and $\Theta\in\mc H$. Thus, in this case, the DS hierarchy~\eqref{DS:hier} is bi-Hamiltonian.

\subsubsection{Hamiltonian densities}\label{sec:Ham1}

As in Section~\ref{section1}, let $U\in\mc A^q\otimes(\im\ad\Lambda)^{<0}$ and
$H\in\mc A^q\otimes\mc H^{<0}$ be such that \eqref{eq:L0} holds, namely ${\rm e}^{\ad U}\mc L=\partial +\Lambda+H$.
Recall from Section~\ref{sec:km} that the DS hierarchy~\eqref{DS:hier} (cf.~\eqref{DShierarinham2})
can be written in Hamiltonian form with respect to the
Poisson structure \eqref{eq:Poisson1} and the Hamiltonian
\begin{equation}\label{eq:Ham-functional}
\bar h_{a,k} := \int\kappa(\Lambda_{m_a+rhk},H)\,{\rm d}x, \qquad a=1,\dots,n,\quad k\geq0.
\end{equation}
Let us recall the following result from~\cite{DS85}.

\begin{Proposition}[\cite{DS85}]
The elements $\bar h_{a,k}$ all belong to $\mc F$, that is, there exist differential polynomials $g_{a,k}\in\mathcal{A}^q$ such that
$h_{a,k}-\p_x(g_{a,k})\in \mc R$.
\end{Proposition}
\begin{proof}
Let $N\in\mc A^q\otimes\mf n$ and $\widetilde{\mc L}={\rm e}^{\ad N}\mc L=\partial+\Lambda+\tilde q$ be as in \eqref{gauge}.
Let
$U_N\in\mc A^q\otimes(\im\ad\Lambda)^{<0}$ and
$\widetilde{H}\in\mc A^q\otimes\mc H^{<0}$ be such that
${\rm e}^{\ad U_N}\widetilde{\mc L}=\partial +\Lambda+\widetilde{H}$.
Then, ${\rm e}^{\ad U_N}{\rm e}^{\ad N}\mc L=\partial +\Lambda+\widetilde{H}$.
Since $\mf n\subset\mf g^{<0}$, by the Baker--Campbell--Haussdorff formula ${\rm e}^{\ad U_N}{\rm e}^{\ad N}={\rm e}^{\ad \widetilde{U}}$ for some
$\widetilde{U}\in\mc A^q\otimes\mf g^{<0}$. Hence, $\widetilde{U}$ and~$\widetilde{H}$ is another solution to equation~\eqref{eq:L0}, namely, ${\rm e}^{\ad \widetilde{U}}\mc L=\partial+\Lambda+\widetilde{H}$.
Recall from~\cite{DS85} that this implies that there exists $S\in\mc A^q\otimes\mc H^{<0}$ such that
${\rm e}^{\ad\widetilde{U}}={\rm e}^{\ad S}{\rm e}^{\ad U_N}$.
Hence,
\begin{align}
H-\widetilde{H}&=\big(1\!-{\rm e}^{\ad S}\big){\rm e}^{\ad U}(\partial +q\!+\Lambda)
=\big(1\!-{\rm e}^{\ad S}\big)(\partial \!+\Lambda+H) =-\!\sum_{k\geq1}\frac{(\ad S)^k}{k!}(\partial+\Lambda\!+H)\nn
\\
&=\partial S
-\sum_{k\geq2}\frac{(\ad S)^{k-1}}{k!}(\partial S)
-\sum_{k\geq1}\frac{(\ad S)^k}{k!}(\Lambda+H)
. \label{20210608:eq1}
\end{align}
Let $\Lambda_i\in\mc H$, $i=m_a+rhk$. By pairing both sides of \eqref{20210608:eq1} with~$\Lambda_i$,
using the invariance of the bilinear form and the fact that $\mc H$ is
abelian we get
\begin{equation}\label{20210608:eq2}
\kappa(\Lambda_i,H)=\kappa\big(\Lambda_i,\widetilde{H}\big)+\partial_x \kappa(\Lambda_i,S).
\end{equation}
Equation \eqref{20210608:eq2} implies
that the densities
$h_{a,k}\in\mc A^q$, $a=1,\dots,n$, $k\geq0$, are gauge invariant up to total $x$-derivatives, hence
$\bar h_{a,k}\in\mc F$.
\end{proof}

As in Section~\ref{sec:matrix_res}, let us take the standard realization of~$\widetilde{\mf g}$ corresponding to the vertex $c_m$.
The $\mb C$-valued bilinear form~$\kappa$ on~$\widetilde{\mf g}$, coordinated with the principal and standard gradations,
can be realized as follows:
\begin{equation}\label{eq:kappa}
\kappa(a\otimes f(\lambda),b\otimes g(\lambda))
=\Res_\lambda (a\otimes f(\lambda)|b\otimes g(\lambda))\lambda^{-1}
=(a|b)\Res_\lambda f(\lambda)g(\lambda)\lambda^{-1}.
\end{equation}
Noting that $\mc H^{<0}$ is spanned by $\Lambda_{m_a-rhk}(\lambda)$, $a=1,\dots,n$, $k\geq0$,
and using~\eqref{eq:normalization1}, we rewrite $H=H(\lambda)\in\mc A^q\otimes\mc H^{<0}$
more explicitly as
\begin{equation}\label{eq:H}
H(\lambda)=\sum_{b=1}^n\sum_{l\in\mb Z_{\geq0}} H_{b,l}\Lambda_{m_b-rh(l+1)}(\lambda)
=\sum_{b=1}^n\sum_{l\in\mb Z_{\geq0}} H_{b,l}\Lambda_{m_b}(\lambda)\lambda^{-(l+1)N_m},
\end{equation}
where $H_{b,l}\in\mc A^q$.
Recalling the definition of the bilinear form $\kappa$ in \eqref{eq:kappa} and $\bar h_{a,k}\in\mc F$ in \eqref{eq:Ham-functional}, we have, for every $a=1,\dots,n$ and $k\geq0$:
{\samepage\begin{align}
\bar h_{a,k}&=\int \kappa(\Lambda_{m_a+rhk}(\lambda),H(\lambda))
=\int \Res_{\lambda}\sum_{b=1}^n\sum_{l\in\mb Z_{\geq0}}(\Lambda_{m_a}(\lambda)|\Lambda_{m_b}(\lambda))H_{b,l}\lambda^{(k-l-1)N_m-1}\nn
\\
&=\int \Res_\lambda\sum_{l\in\mb Z_{\geq0}} h H_{n+1-a,l}\lambda^{(k-l)N_m-1}=\int h H_{n+1-a,k},
\label{20210324:eq1}
\end{align}
where in the third equality we used equation \eqref{eq:normalization}.}

Let us collect the densities 
$h_{a,k}$, $a=1,\dots,n$, $k\geq0$, into $n$ generating series using the $\mb C\big(\big(\lambda^{-1}\big)\big)$-valued bilinear form
on $\widetilde{\mf g}$ (see \cite{BDY16, BDY21} and \cite{DSKV13} for more details on the analogous construction for the untwisted case)
by letting
\begin{equation}\label{eq:g}
g_a(\lambda)=(\Lambda_{m_a}(\lambda)|H(\lambda))\in\mc A^q\big(\big(\lambda^{-1}\big)\big),\qquad
a=1,\dots,n.
\end{equation}
Indeed, using equations \eqref{eq:H}, \eqref{eq:normalization} and \eqref{20210324:eq1} we get
\begin{equation}\label{eq:ga_explicitly}
g_a(\lambda)=\sum_{k\in\mb Z_{\geq0}} h H_{n+1-a,k}\lambda^{-kN_m}
=\sum_{k\in\mb Z_{\geq0}}h_{a,k}\lambda^{-kN_m}
\in\mc A^q\big[\big[\lambda^{-N_m}\big]\big].
\end{equation}

\subsection[The series G\_a(lambda)]{The series $\boldsymbol{G_a(\lambda)}$}

Recall from equation \eqref{eq:G} the gauge-invariant differential polynomials $\Omega_{a,k;1,0}$
and its generating series $G_a(\lambda)$.
In this subsection we establish a relation between $\Omega_{a,k;1,0}$ and~$h_{a,k}$ by
 deriving an identity between the series $G_a(\lambda)$ and $g_a(\lambda)$. To proceed,
we need the following results.
\begin{Lemma}[\cite{DSKV13}] \phantomsection\label{lemma2}\quad
\begin{enumerate}[$(a)$]
\item
Let $D$ be a derivation of $\mc A^q\otimes\widetilde{\mf g}$. For every $\alpha\in\mb C$, $A,U_1,\dots,
U_k\in\mc A^q\otimes \widetilde{\mf g}$, with $k\geq1$, we have
\begin{align}\label{eq:induttiva}
D \big(\ad U_1\cdots\ad U_k(\alpha\partial_x+A)\big)
={}&\sum_{h=1}^k\ad U_1\cdots\ad D( U_h)\cdots\ad U_k(\alpha\partial_x+A)\nn
\\
&+\ad U_1\cdots\ad U_k( D(A))\nn
\\
&-\alpha\ad U_1\cdots\ad U_{k-1}([D,\partial_x](U_k)).
\end{align}
\item
For any $\alpha\in\mb C$ and $A,U,V\in\mc A^q\otimes \widetilde{\mf g}$ we have
\begin{gather*}
\bigg[\sum_{h\in\mb Z_{\geq0}}\frac{1}{(h+1)!}(\ad U)^h(V),{\rm e}^{\ad U}(\alpha\partial_x +A)\bigg]
\\ \qquad
{}=\sum_{h,k\in\mb Z_{\geq0}}\frac{1}{(h+k+1)!}(\ad U)^h\ad V(\ad U)^k(\alpha\partial_x+A).
\end{gather*}
\end{enumerate}
\end{Lemma}

Now we can prove the following proposition.
\begin{Proposition}\label{prop1}
We have that
\[
G_a(\lambda)=\left(\partial_\lambda-\frac{m_aN_m}{rh\lambda}\right)g_a(\lambda)+\partial_x X(\lambda),
\]
for some $X(\lambda)\in\mc A^q\big[\big[\lambda^{-N_m}\big]\big]\lambda^{-1}$.
\end{Proposition}

\begin{proof}Let us start by computing $\partial_\lambda H(\lambda)$, where $H(\lambda)$ is the series appearing in~\eqref{eq:L0}. We have
\begin{align}
\partial_\lambda H(\lambda)&=\partial_\lambda\big( {\rm e}^{-\ad U(\lambda)}(\mc L)-\partial_x-\Lambda(\lambda)\big)
=\sum_{k=1}^{\infty}\frac{(-1)^k}{k!}\partial_\lambda\big((\ad U(\lambda))^k(\mc L)\big)\notag
\\
&=\sum_{k=1}^{\infty}\sum_{h=0}^{k-1}\frac{(-1)^k}{k!}(\ad U(\lambda))^h(\ad \partial_\lambda U(\lambda))(\ad U(\lambda))^{k-1-h}(\mc L)\nn
\\
&\phantom{=}+\sum_{k=1}^\infty\frac{(-1)^k}{k!}(\ad U(\lambda))^k(\tilde e_m)\notag\\
&=\sum_{h,k\in\mb Z_{\geq0}}\frac{(-1)^{h+k+1}}{(h+k+1)!}(\ad U(\lambda))^h(\ad \partial_\lambda U(\lambda))(\ad U(\lambda))^{k}(\mc L)
+{\rm e}^{-\ad U(\lambda)}(\tilde e_m)-\tilde e_m
\notag\\
&=[{\rm e}^{-\ad U(\lambda)}(\mc L), Y(\lambda)]+{\rm e}^{-\ad U(\lambda)}(\tilde e_m)-\tilde e_m\notag\\
&=\partial_x Y(\lambda)+[\Lambda(\lambda)+H(\lambda),Y(\lambda)]+{\rm e}^{-\ad U(\lambda)}(\tilde e_m)-\tilde e_m
,\label{20200828:eq4}
\end{align}
where $Y(\lambda)=\sum_{h\in\mb Z_{\geq0}}\frac{(-1)^h}{(h+1)!}(\ad U(\lambda))^h(\partial_\lambda U(\lambda))$.
In the first equality we used equation \eqref{eq:L0}, in the second equality we used the definition of exponential function,
in the third equality equation \eqref{eq:induttiva} and the facts that $\partial_\lambda\mc L=\tilde e_m$ and $[\partial_\lambda,\partial_x]=0$,
the fourth equality is trivial, in the fifth equality we used Lemma \ref{lemma2}(b), and in the last equality we used again
equation \eqref{eq:L0}.

From the definition of $g_a(\lambda)$ we then have
\begin{equation}\label{20200828:eq5}
\lambda\partial_\lambda g_a(\lambda)
=(\lambda\partial_\lambda \Lambda_{m_a}(\lambda)|H(\lambda))+(\Lambda_{m_a}(\lambda) |\lambda \partial_\lambda H(\lambda))
.
\end{equation}
Since $H(\lambda)\in\mc A^q\otimes\Ker\ad \Lambda(\lambda)$ and the direct sum decomposition \eqref{dec:Lambda} is orthogonal
with respect to $(\cdot\,|\,\cdot)$ we have, using equation \eqref{20200828:eq3},
\begin{align}
(\lambda \partial_\lambda \Lambda_{m_a}(\lambda) | H(\lambda))
&=(\pi_{\mc H}(\lambda\partial_\lambda \Lambda_{m_a}(\lambda)) | H(\lambda))
=\frac{m_aN_m}{rh}(\Lambda_{m_a}(\lambda)|H(\lambda))\nn
\\
&=\frac{m_aN_m}{rh}g_a(\lambda).\label{20200828:eq6}
\end{align}
Furthermore, using equation \eqref{20200828:eq4} we obtain
\begin{align}
(\Lambda_{m_a}(\lambda)|\partial_\lambda H)
&=(\Lambda_{m_a}(\lambda)|\partial_x Y(\lambda) )+(\Lambda_{m_a}(\lambda)|[\Lambda(\lambda)+H,Y(\lambda)])
+(\Lambda_{m_a}(\lambda)|{\rm e}^{-\ad U}(\tilde e_m))
\nn
\\
&\phantom{=}-(\Lambda_{m_a}(\lambda)|\tilde e_m)
=(\Lambda_{m_a}(\lambda) | \partial_xY(\lambda) )+([\Lambda_{m_a}(\lambda),\Lambda(\lambda)+H(\lambda)]|Y(\lambda))\nn
\\
&\phantom{=}+(R_{m_a}(\lambda)|\tilde e_m)
-\delta_{a,n}\frac{N_m}{r}\lambda^{N_m-1}\nn
\\
&=\partial_x(\Lambda_{m_a}(\lambda)|Y(\lambda) )+(R_{m_a}(\lambda)|\tilde e_m)-\delta_{a,n}\frac{N_m}{r}\lambda^{N_m-1}.\label{20200828:eq7}
\end{align}
Here, in the second equality we used the invariance of the Cartan--Killing form, the definition of~$R_{m_a}$ and
equation~\eqref{20210128:eq1b}; in the third equality we used the facts that
$\partial_x\Lambda_{m_a}(\lambda)=0$ and $\mc H$ is abelian.
Combining equations \eqref{20200828:eq5}, \eqref{20200828:eq6} and \eqref{20200828:eq7} we get the identity
\begin{equation}\label{20210128:eq1c}
(R_{m_a}(\lambda)|\tilde e_m)=\bigg(\partial_\lambda-\frac{m_aN_m}{rh
\lambda}\bigg)g_a(\lambda)-\partial_x (\Lambda_{m_a}(\lambda)|Y(\lambda))
+\delta_{a,n}N_m\lambda^{N_m-1}.
\end{equation}
Recall from equation \eqref{eq:ga_explicitly} that $g_a(\lambda)\in\mc A^q\big[\big[\lambda^{-N_m}\big]\big]$. Hence,
the claim follows by applying $\pi_\lambda$ to both sides of~\eqref{20210128:eq1c}.
\end{proof}

\begin{Corollary} \label{20210328cor1}
The elements $\Omega_{a,k;1,0}$ are related to~$h_{a,k}$ by
\[
\Omega_{a,k;1,0}=-(m_a+rhk)\frac{N_m}{rh}h_{a,k}+ \partial_x X_{a,k}
\]
for some $X_{a,k}\in\mc A^q$. In other words,
the DS hierarchy~\eqref{DS:hier} can be written in terms of the gauge invariant densities~$\Omega_{a,k;1,0}$ as follows:
\begin{equation}\label{20210328:eq1}
\frac{\partial u_s}{\partial t_{m_a+krh}} = \frac{-rh}{(m_a+rhk) N_m} \bigl\{\overline{\Omega}_{a,k;1,0}, \, u_s(x)\bigr\}_2.
\end{equation}
\end{Corollary}
\begin{proof}
Follows straightforwardly from Proposition~\ref{prop1}.
\end{proof}

\begin{Example}
Let~$L(\partial)$ be the Lax operator of the SK hierarchy
given in Section~\ref{sec:LaxSK} (respectively of the KK hierarchy given in Section~\ref{subsubsectionkk}).
We have (see \cite{DSKV18})
\begin{equation}\label{eq:remark}
h_{a,k}=-\frac{m_a+6k}{3}\Res_{\partial}L^{\frac{m_a+6k}{3}}(\partial)+\partial_x X_{a,k},\qquad
a=1,5,\quad k\geq0,
\end{equation}
for some $X_{a,k}\in\mc A^q$.
The identity \eqref{eq:remark} agrees with Corollary \ref{20210328cor1}
and Proposition~\ref{prop:SK} (respectively Proposition~\ref{prop:KK}).
\end{Example}

In the last part of this section we will prove the following theorem, which extends the
result in~\cite{BDY16} (cf.~also~\cite{DYZ21})
for the $A_1^{(1)}$ case.

\begin{Theorem}\label{thmunt}
For an untwisted affine Kac--Moody algebra and the choice of the special vertex~$c_0$ of its Dynkin diagram,
under a suitable Miura-type transformation,
the gauge invariant differential polynomials $\Omega_{a,k;1,0}$ satisfy the axioms of tau-symmetric bi-Hamiltonian structure given in {\rm \cite{DZ-norm}}.
\end{Theorem}

\begin{proof}
It is assumed here that $r=1$ (untwisted case) and that we choose the special vertex~$c_0$. In~particular,
we have that $n=\ell$ and $N_0=1$ (see Section~\ref{sec:realizationKM}).
Moreover, recall from Section~\ref{sec:km} that in this case the DS hierarchy is bi-Hamiltonian
and we can rewrite equation~\eqref{20210328:eq1} in Corollary~\ref{20210328cor1}
using the first Hamiltonian structure as follows (cf.~\eqref{eq:normalization} and~\eqref{firsthamDS})
\begin{gather}\label{20210328:eq1b}
\frac{\partial u_s}{\partial t_{m_a+kh}}
= \frac{-h}{m_a+h(k+1)} \bigl\{\overline{\Omega}_{a,k+1;1,0}, \, u_s(x)\bigr\}_1.
\end{gather}
This means that $\Omega_{a,k+1,0} \in \R$ are (up to a scalar factor) Hamiltonian densities
associated to the first Hamiltonian structure~\eqref{eq:Poisson2b} for the DS hierarchy.

Next, from Corollary~\ref{20210328cor1}
and from the fact that $\bar h_{a,0}$ are Casimirs of the first Hamiltonian structure proved in~\cite{DSKV13}
we know that $r_a=\Omega_{a,0;1,0}$ are densities of these Casimirs.
Moreover, as it was shown in \cite[Lemma 4.1.3]{BDY21}, the map $(u_1,\dots,u_n)\mapsto (r_1,\dots,r_n)$
is a Miura-type transformation meaning that the dispersionless limit of this map has non-degenerate Jacobian.

On the other hand, it is shown in~\cite{DLZ} that the genus zero part of the DS hierarchy is equivalent to the
principal Hierarchy of the Frobenius manifold of type~$\mf a$~\cite{D96}.
Finally, using~\eqref{tau09081}, \eqref{tau09082} and the normalization of times given in~\cite{BDY21} (see $t^{a,k}$ therein), it follows that
under the Miura transformation $(u_1,\dots,u_n)\mapsto (r_1,\dots,r_n)$
(see, e.g.,~\cite{DLYZ16})
the DS hierarchy written in the coordinate~$r_a$ (cf.~\eqref{20210328:eq1b})
\begin{equation}\label{dsr1010}
\frac{\p r_a}{\p t_{m_b+krh}} = \left\{\bar h_{b,k+1},r_a(x)\right\}_1 = -\p_x \left(\Omega_{b,k;a,0}\right)
\end{equation}
 is a tau-symmetric bi-Hamiltonian deformation
of the principal Hierarchy in the sense of~\cite{DZ-norm} (cf.~also~\cite{DLYZ16}),
and the differential polynomials $\Omega_{a,k;1,0}$ are tau-symmetric Hamiltonian densities
in the sense of~\cite{DZ-norm}.
\end{proof}
According to Theorem~\ref{thmunt}, the tau-structure $(\Omega_{a,k;b,l})_{k,l\geq0, a,b=1,\dots,n}$
coincides with the axiomatic tau-structure in~\cite{DZ-norm} for the DS hierarchies under the assumption of
Theorem~\ref{thmunt}.
The coordinates~$r_a$ introduced in the above proof are called {\it normal coordinates}~\cite{DLYZ16,DZ-norm}.

\begin{Remark}
Under the condition of Theorem~\ref{thmunt},
one easily sees from~\eqref{dsr1010} (choosing $a=1$) and a homogeneity argument
that the tau-symmetric hamiltonian densities for the DS hierarchy if exist must be unique.
The construction of tau-symmetric hamiltonian densities for DS hierarchies was previously given in~\cite{Wu17} (cf.~also~\cite{LRZ15,LWZZ20})
by using the central extension of $\widetilde{\mf g}$.
Our construction, however, uses only the geometry of the resolvent manifold
$\mc M_{\mc L}=\{R\in\mc A^q\otimes\widetilde{\mf g}\mid [R,\mc L]=0\}$~\cite{DS85}
and is therefore simpler from the computational point of~view.
\end{Remark}

\subsection[Variational derivative of the series G\_a(lambda)]{Variational derivative of the series $\boldsymbol{G_a(\lambda)}$} 

Recall from Section~\ref{sec:km} that $(v_i)_{i=1,\dots,\dim \mf b}$ is a
basis of $\mf b$ and $(v^i)_{i=1,\dots,\dim \mf b}$ is the dual basis of $\mf a^{\geq0}$.
Moreover, the basis is chosen so that $(v_i)_{i=1,\dots,\ell}$ is a basis of the Drinfeld--Sokolov gauge $V\subset\mf b$
and $(v_i)_{i=\ell+1,\dots,\dim \mf b}$ is a basis of $[e,\mf n]$. Recall from Section~\ref{sec:km} that $(v^i)_{i=1,\dots,\ell}$ is a basis for $\mf a^e$.

For $A(\lambda)\in \widetilde{\mf g}$ we define its projection on $\tilde{\mf a}^{\geq0}\! :=\mf a^{\geq0}\big(\big(\lambda^{-N_m}\big)\big)$
\big(respectively $\tilde{\mf a}^{e}\! :=\mf a^{e}\big(\big(\lambda^{-N_m}\big)\big)$\big) by
\begin{gather}\label{20210610:eq1}
\pi_{\tilde{\mf a}^{\geq0}}(A(\lambda))=\sum_{i=1}^{\dim \mf b} (A(\lambda)|v_i)v^i
\qquad\bigg(\text{respectively }\pi_{\tilde{\mf a}^{e}}(A(\lambda))=\sum_{i=1}^\ell(A(\lambda)|v_i)v^i\bigg).
\end{gather}

Recall the definition of variational derivatives given in \eqref{eq:varder} and \eqref{eq:vardercan}.
With a similar method used in~\cite{DSKV13}
we can prove that
\[
\frac{\delta \bar g_a(\lambda)}{\delta q}=\pi_{\tilde{\mf a}^{\geq0}}R_{m_a}(\lambda),\qquad
a=1,\dots,n.
\]
Moreover, in a similar way for which we are going to provide the details for completeness,
we can prove the following result.
\begin{Proposition}\label{20210610:prop1}
For every $a=1,\dots,n$ we have
\begin{equation*} 
\frac{\delta \bar g_a(\lambda)}{\delta q^{\rm can}}
=\pi_{\tilde{\mf a}^e}\big({\rm e}^{\ad N^{\rm can}}R_{m_a}(\lambda)\big)\in\mc \R \otimes\mf a^e\big[\big[\lambda^{-N_m}\big]\big],
\end{equation*}
where $N^{\rm can}\in\mc A^q\otimes \mf n$ is defined by \eqref{Lcan}.
\end{Proposition}
\begin{proof}
Let $U^{\rm can}(\lambda)\in\mc \R \otimes(\im\ad \Lambda(\lambda))^{<0}$ and $H^{\rm can}(\lambda)\in\mc \R \otimes\mc H_{<0}$
be the unique elements such that
\begin{equation}\label{L0_ds}
{\rm e}^{\ad U^{\rm can}(\lambda)}(\partial+\Lambda(\lambda)+q^{\rm can})=\partial+\Lambda(\lambda)+H^{\rm can}(\lambda).
\end{equation}
By the definition \eqref{eq:vardercan} of the variational derivative,
the definition \eqref{eq:g} of $\bar g_a(\lambda)$ and equation~\eqref{20210608:eq2}, we have
\begin{align}
\frac{\delta \bar g_a(\lambda)}{\delta q^{\rm can}}
&=\sum_{i=1}^\ell\sum_{m\geq0}(-\partial)^m
\frac{\partial g_a(\lambda)}{\partial u_{i,mx}}v^i
=\sum_{i=1}^\ell\sum_{m\geq0}(-\partial)^m
\bigg(\Lambda_{m_a}(\lambda)\,\bigg|\,\frac{\partial H^{\rm can}(\lambda)}{\partial u_{i,mx}}\bigg)v^i\nn
\\
&=\sum_{i=1}^\ell\sum_{m\geq0}(-\partial)^m
\big(\Lambda_{m_a}(\lambda)\mid\frac{\partial}{\partial u_{i,mx}}
\big({\rm e}^{\ad U^{\rm can}(\lambda)} (\mc L^{\rm can})-\partial-\Lambda(\lambda)
\big)\big) v^i.\label{eq:30apr-1}
\end{align}
In the last identity we used equation \eqref{L0_ds}.
We next expand ${\rm e}^{\ad U^{\rm can}(\lambda)}$ in power series.
Since $q^{\rm can}=\sum_{i=1}^\ell u_i v_i$ and using the definition \eqref{20210610:eq1} of $\pi_{\tilde{\mf a}^e}$
we find that the first term of the expansion~is
\begin{equation}\label{eq:30apr-1a}
\sum_{i=1}^\ell\sum_{m\geq0}(-\partial)^m
\bigg(\Lambda_{m_a}(\lambda)\,\bigg|\,\frac{\partial q^{\rm can}}{\partial u_{i,mx}}\bigg)
=\sum_{i=1}^\ell(\Lambda(\lambda)|v_i)v^i=
\pi_{\tilde{\mf a}^e} \Lambda_{m_a}(\lambda).
\end{equation}
By Lemma \ref{lemma2}, all the other terms in the power series expansion of the RHS of \eqref{eq:30apr-1} are
\begin{gather}
\sum_{k=1}^\infty\frac1{k!}\sum_{i=1}^\ell\sum_{m\geq0}
(-\partial)^m\bigg(\Lambda_{m_a}(\lambda)\,\bigg|\,
\frac{\partial}{\partial u_{i,mx}}(\ad U^{\rm can}(\lambda))^k\mc L^{\rm can}\bigg)\nn
\\ \qquad
{}=\sum_{k=1}^\infty\frac1{k!}\sum_{i=1}^\ell\sum_{m\geq0}
(-\partial)^m\Bigg(\Lambda_{m_a}(\lambda)\,\Bigg|\,
\sum_{h=0}^{k-1}(\ad U^{\rm can}(\lambda))^h \bigg(\ad \frac{\partial U^{\rm can}(\lambda)}{\partial u_{i,mx}}\bigg) \nn
\\ \qquad\phantom{=}
{}\times(\ad U^{\rm can}(\lambda))^{k-h-1} \mc L^{\rm can} \nn
\\ \qquad\phantom{=}
{}+(\ad U^{\rm can}(\lambda))^k\frac{\partial}{\partial u_{i,mx}}(q^{\rm can}+\Lambda(\lambda))
-(\ad U^{\rm can}(\lambda))^{k-1}\frac{\partial U^{\rm can}(\lambda)}{\partial u_i^{(m-1)}}
\Bigg)\nn
\\ \qquad
{}=\sum_{h,k\geq0}^\infty\frac1{(h+k+1)!}
\sum_{i=1}^\ell\sum_{m\geq0}(-\partial)^m\bigg(\Lambda_{m_a}(\lambda)\,\bigg|\,
(\ad U^{\rm can}(\lambda))^h \biggl(\ad \frac{\partial U^{\rm can}(\lambda)}{\partial u_{i,mx}}\biggr) \nn
\\ \qquad\phantom{=}
\times(\ad U^{\rm can}(\lambda))^k \mc L^{\rm can}\bigg)
+\sum_{k=1}^\infty\frac1{k!}
\pi_{\tilde{\mf a}^e}\big((-\ad U^{\rm can}(\lambda))^k (\Lambda_{m_a}(\lambda))\big)\nn
\\ \qquad
{}-\sum_{k\geq0}\frac1{(k+1)!}\sum_{i=1}^\ell\sum_{m\geq0}(-\partial)^m
\bigg(\Lambda_{m_a}(\lambda)\,\bigg|\,
(\ad U^{\rm can}(\lambda))^{k}\frac{\partial U^{\rm can}(\lambda)}{\partial u_i^{(m-1)}}\bigg).
\label{eq:30apr-1b}
\end{gather}
For the first and last terms in the RHS we just changed the summation indices,
while for the second term we used the definition~\eqref{20210610:eq1} of the map $\pi_{\tilde{\mf a}^e}$
and the invariance of the bilinear form.
Combining~\eqref{eq:30apr-1a} and the second term in the RHS of~\eqref{eq:30apr-1b},
we get
$\pi_{\tilde{\mf a}^e}\big({\rm e}^{-\ad U^{\rm can}}(\Lambda_{m_a}(\lambda))\big)=\pi_{\tilde{\mf a}^e}\big(R_{m_a}^{\rm can}(\lambda)\big)$,
where we recall that $R_{m_a}^{\rm can}$ is defined in~\eqref{RcanNR}.
Hence, in order to complete the proof of the proposition,
we are left to show that the first and last term in the RHS of~\eqref{eq:30apr-1b} cancel out.
The last term of the RHS of~\eqref{eq:30apr-1b} can be rewritten as
\begin{equation}\label{eq:30apr-3}
-\sum_{i=1}^\ell\sum_{m\geq1}
(-\partial)^m(\Lambda_{m_a}(\lambda)\mid
X_{i,m-1}(\lambda)),
\end{equation}
where $X_{i,m}(\lambda)=\sum_{k\geq0}\frac1{(k+1)!}
(\ad U^{\rm can}(\lambda))^{k}\frac{\partial U^{\rm can}(\lambda)}{\partial u_{i,mx}}$.
On the other hand, by Lemma~\ref{lemma2}b), the first term of the RHS of~\eqref{eq:30apr-1b} is equal to
\[
\sum_{i=1}^\ell\sum_{m\geq0}(-\partial)^m\bigl(\Lambda_{m_a}(\lambda) \mid
\big[X_{i,m}(\lambda),{\rm e}^{\ad U^{\rm can}(\lambda)}\mc L^{\rm can}\big]\bigr).
\]
By equation \eqref{L0_ds}, the invariance of the bilinear form and the fact that $\mc H$ is abelian
the above expression is equal to
\[
\sum_{i=1}^\ell\sum_{m\geq0}
(-\partial)^{m+1}(\Lambda_{m_a}(\lambda)\mid X_{i,m}(\lambda)),
\]
which, combined with \eqref{eq:30apr-3}, gives zero.
The fact that
$\pi_{\tilde{\mf a}^e}\left({\rm e}^{\ad N^{\rm can}}R_{m_a}(\lambda)\right)\in \R \otimes\mf a^e\big[\big[\lambda^{-N_m}\big]\big]$
follows by simple degree considerations.
\end{proof}
\begin{Corollary}
For every $a=1,\dots,n$ we have
\begin{equation} \label{var_der_Ga}
\frac{\delta \overline{G}_a(\lambda)}{\delta q^{\rm can}}
=\bigg(\partial_\lambda-\frac{m_aN_m}{rh\lambda}\bigg)
\big(\pi_{\tilde{\mf a}^e}\big({\rm e}^{\ad N^{\rm can}}R_{m_a}(\lambda)\big)\big)\in \R\otimes\mf a^e\big[\big[\lambda^{-N_m}\big]\big]\lambda^{-1},
\end{equation}
where $N^{\rm can}\in\mc A^q\otimes \mf n$ is defined by \eqref{Lcan}.
\end{Corollary}
\begin{proof}
Immediate from Propositions \ref{prop1} and \ref{20210610:prop1} and the fact that total derivatives are in the kernel of the operator of variational derivative.
\end{proof}

Let us now assume that $r=1$ and let us choose the special vertex~$c_0$ of the Kac--Moody algebra,
so that in this case $\widetilde{\mf g}=\mf g\big(\big(\lambda^{-1}\big)\big)$ for a simple Lie algebra $\mf g$
of rank $n$ (we also have $\ell=n$). It is possible to choose a basis $\{v_i\}_{i=1}^n$ of $V$
such that $v_n=e_{-\theta}$ (up to a constant multiple), where $\theta$ is the highest root of $\mf g$. Moreover, $\tilde e_m=e_{-\theta}$, namely $\Lambda(\lambda)=e+e_{-\theta}\lambda$.
In~this setting, we have the following corollary.
\begin{Corollary}\label{20210610:cor2}
For an untwisted affine Kac--Moody algebra with the choice of the special vertex~$c_0$, the following identity holds:
\begin{equation*} 
\frac{\delta \overline{G}_a(\lambda)}{\delta u_n}
=\bigg(\partial_\lambda-\frac{m_a}{h\lambda}\bigg) G_a(\lambda) \in \R\big[\big[\lambda^{-N_m}\big]\big]\lambda^{-1}.
\end{equation*}
\end{Corollary}
\begin{proof}
From the definition of the variational derivative \eqref{eq:vardercan} and the fact that $v_n=e_{-\theta}$, we have
\[
\frac{\delta \overline{G}_a(\lambda)}{\delta u_n}=\bigg(\frac{\delta \overline{G}_a(\lambda)}{\delta q^{\rm can}}\,\bigg|\,v_n\bigg)
=\bigg(\frac{\delta \overline{G}_a(\lambda)}{\delta q^{\rm can}}\,\bigg|\,e_{-\theta}\bigg).
\]
Hence, by equation \eqref{var_der_Ga} we have
\begin{align*}
\frac{\delta \overline{G}_a(\lambda)}{\delta u_n}&=
\bigg(\partial_\lambda-\frac{m_aN_m}{rh\lambda}\bigg)
\big(\pi_{\tilde{\mf a}^e}\big({\rm e}^{\ad N^{\rm can}}R_{m_a}(\lambda)\big)\mid e_{-\theta}\big)
\\
&=\bigg(\partial_\lambda-\frac{m_aN_m}{rh\lambda}\bigg)
\pi_\lambda\big(R_{m_a}(\lambda)\mid {\rm e}^{-\ad N^{\rm can}}(e_{-\theta})\big)
\\
&=\bigg(\partial_\lambda-\frac{m_a N_m}{rh\lambda}\bigg)
\pi_\lambda\big(R_{m_a}(\lambda)\mid e_{-\theta}\big)
=\bigg(\partial_\lambda-\frac{m_aN_m}{rh\lambda}\bigg)G_a(\lambda).
\end{align*}
In the second equation we used the fact the bilinear form is coordinated with the gradation and its invariance, in the third equality
we used that fact that $[\mf n,e_{-\theta}]=0$, and finally we used the definition~\eqref{eq:G} of~$G_a(\lambda)$ and the fact that
$\tilde e_m=e_{-\theta}$.
\end{proof}

\begin{Remark}
We note that Corollary~\ref{20210610:cor2} and the criterion in~\cite{BDGR18}
lead to another proof of Theorem~\ref{thmunt}.
\end{Remark}

\subsection*{Acknowledgements}

Part of the work of D.V.~and D.Y.~was done during their visits to SISSA and Tsinghua University during the years 2017 and 2018; they thank both
SISSA and Tsinghua for warm hospitality and financial support.
D.V.~acknowledges the financial support of the
project MMNLP (Mathematical Methods in Non Linear Physics) of the
INFN.
The work of D.Y.~was partially supported by the National Key R and D Program of China 2020YFA0713100, and by NSFC 12061131014.

\subsection*{Note added}
The collaborative research of our project aiming at generalizing
 the results in~\cite{BDY21} to twisted affine Kac--Moody algebras
started in 2017.
The three authors communicated by email and also in person during several visits
of D.V.~and D.Y.~to Trieste to meet with B.D.~at SISSA.
During these periods, we achieved the extension of the matrix-resolvent method to
the DS hierarchies associated to affine Kac--Moody algebras,
and a draft containing the main results of what are now Sections \ref{section1}--\ref{sec:exa} was written by the three of us,
while Section~\ref{sec:Ham} contains further results found
by D.V.~and D.Y.~after Boris Dubrovin passed away in March of 2019.


\begin{thebibliography}{99}
\footnotesize\itemsep=0pt

\bibitem{BFORFW}
Balog J., Feh\'er L., O'Raifeartaigh L., Forg\'acs P., Wipf A., Toda theory and
 {$\mathcal{W}$}-algebra from a gauged {WZNW} point of view, \href{https://doi.org/10.1016/0003-4916(90)90029-N}{\textit{Ann.
 Physics}} \textbf{203} (1990), 76--136.

\bibitem{BDY16}
Bertola M., Dubrovin B., Yang D., Correlation functions of the {K}d{V}
 hierarchy and applications to intersection numbers over
 {$\overline{\mathcal{M}}_{g,n}$}, \href{https://doi.org/10.1016/j.physd.2016.04.008}{\textit{Phys.~D}} \textbf{327} (2016),
 30--57, \href{https://arxiv.org/abs/1504.06452}{arXiv:1504.06452}.

\bibitem{BDY21}
Bertola M., Dubrovin B., Yang D., Simple {L}ie algebras, {D}rinfeld--{S}okolov
 hierarchies, and multi-point correlation functions, \href{https://doi.org/10.17323/1609-4514-2021-21-2-233-270}{\textit{Mosc. Math.~J.}}
 \textbf{21} (2021), 233--270, \href{https://arxiv.org/abs/1610.07534}{arXiv:1610.07534}.

\bibitem{BDGR18}
Buryak A., Dubrovin B., Gu\'er\'e J., Rossi P., Tau-structure for the double
 ramification hierarchies, \href{https://doi.org/10.1007/s00220-018-3235-4}{\textit{Comm. Math. Phys.}} \textbf{363} (2018),
 191--260, \href{https://arxiv.org/abs/1602.05423}{arXiv:1602.05423}.

\bibitem{CW19}
Cafasso M., Wu C.-Z., Borodin--{O}kounkov formula, string equation and
 topological solutions of {D}rinfeld--{S}okolov hierarchies, \href{https://doi.org/10.1007/s11005-019-01205-8}{\textit{Lett.
 Math. Phys.}} \textbf{109} (2019), 2681--2722, \href{https://arxiv.org/abs/1505.00556}{arXiv:1505.00556}.

\bibitem{CM}
Collingwood D.H., McGovern W.M., Nilpotent orbits in semisimple {L}ie algebras,
 Van Nostrand Reinhold Mathematics Series, Van Nostrand Reinhold Co., New
 York, 1993.

\bibitem{dGHM}
de~Groot M.F., Hollowood T.J., Miramontes J.L., Generalized
 {D}rinfel'd--{S}okolov hierarchies, \href{https://doi.org/10.1007/BF02099281}{\textit{Comm. Math. Phys.}} \textbf{145}
 (1992), 57--84.

\bibitem{DSJKV21}
De~Sole A., Jibladze M., Kac V.G., Valeri D., Integrability of classical affine
 $\mathcal W$-algebras, \href{https://doi.org/10.1007/s00031-021-09645-0}{\textit{Transf. Groups}} \textbf{26} (2021), 479--500,
 \href{https://arxiv.org/abs/2007.01244}{arXiv:2007.01244}.

\bibitem{DSKV13}
De~Sole A., Kac V.G., Valeri D., Classical {$\mathcal W$}-algebras and
 generalized {D}rinfeld--{S}okolov bi-{H}amiltonian systems within the theory
 of {P}oisson vertex algebras, \href{https://doi.org/10.1007/s00220-013-1785-z}{\textit{Comm. Math. Phys.}} \textbf{323} (2013),
 663--711, \href{https://arxiv.org/abs/1207.6286}{arXiv:1207.6286}.

\bibitem{DSKV16}
De~Sole A., Kac V.G., Valeri D., Structure of classical (finite and affine)
 {$\mathcal W$}-algebras, \href{https://doi.org/10.4171/JEMS/632}{\textit{J.~Eur. Math. Soc.}} \textbf{18} (2016),
 1873--1908, \href{https://arxiv.org/abs/1404.0715}{arXiv:1404.0715}.

\bibitem{DSKV18}
De~Sole A., Kac V.G., Valeri D., Classical affine {$\mathcal W$}-algebras and
 the associated integrable {H}amiltonian hierarchies for classical {L}ie
 algebras, \href{https://doi.org/10.1007/s00220-018-3142-8}{\textit{Comm. Math. Phys.}} \textbf{360} (2018), 851--918,
 \href{https://arxiv.org/abs/1705.10103}{arXiv:1705.10103}.

\bibitem{Dinar2}
Dinar Y.I., Frobenius manifolds from subregular classical {$W$}-algebras,
 \href{https://doi.org/10.1093/imrn/rns121}{\textit{Int. Math. Res. Not.}} \textbf{2013} (2013), 2822--2861,
 \href{https://arxiv.org/abs/1108.5445}{arXiv:1108.5445}.

\bibitem{Dinar3}
Dinar Y.I., {$W$}-algebras and the equivalence of bihamiltonian,
 {D}rinfeld--{S}okolov and {D}irac reductions, \href{https://doi.org/10.1016/j.geomphys.2014.06.003}{\textit{J.~Geom. Phys.}}
 \textbf{84} (2014), 30--42, \href{https://arxiv.org/abs/0911.2116}{arXiv:0911.2116}.

\bibitem{DS85}
Drinfel'd V.G., Sokolov V.V., Lie algebras and equations of {K}orteweg--de
 {V}ries type, \href{https://doi.org/10.1007/BF02105860}{\textit{Soviet~J. Math.}} \textbf{30} (1985), 1975--2036.

\bibitem{D96}
Dubrovin B., Geometry of {$2$}{D} topological field theories, in Integrable
 Systems and Quantum Groups ({M}ontecatini {T}erme, 1993), \textit{Lecture Notes in Math.}, Vol.~1620, \href{https://doi.org/10.1007/BFb0094793}{Springer}, Berlin, 1996, 120--348,
 \href{https://arxiv.org/abs/hep-th/9407018}{arXiv:hep-th/9407018}.

\bibitem{DLYZ16}
Dubrovin B., Liu S.-Q., Yang D., Zhang Y., Hodge integrals and tau-symmetric
 integrable hierarchies of {H}amiltonian evolutionary {PDE}s, \href{https://doi.org/10.1016/j.aim.2016.01.018}{\textit{Adv.
 Math.}} \textbf{293} (2016), 382--435, \href{https://arxiv.org/abs/1409.4616}{arXiv:1409.4616}.

\bibitem{DLZ}
Dubrovin B., Liu S.-Q., Zhang Y., Frobenius manifolds and central invariants for
 the {D}rinfeld--{S}okolov bi{H}amiltonian structures, \href{https://doi.org/10.1016/j.aim.2008.06.009}{\textit{Adv. Math.}}
 \textbf{219} (2008), 780--837, \href{https://arxiv.org/abs/0710.3115}{arXiv:0710.3115}.

\bibitem{DYZ21}
Dubrovin B., Yang D., Zagier D., On tau-functions for the {K}d{V} hierarchy,
 \href{https://doi.org/10.1007/s00029-021-00620-x}{\textit{Selecta Math.~(N.S.)}} \textbf{27} (2021), 12, 47~pages,
 \href{https://arxiv.org/abs/1812.08488}{arXiv:1812.08488}.

\bibitem{DZ-norm}
Dubrovin B., Zhang Y., Normal forms of hierarchies of integrable {PDE}s,
 {F}robenius manifolds and {G}romov--{W}itten invariants,
 \href{https://arxiv.org/abs/math.DG/0108160}{arXiv:math.DG/0108160}.

\bibitem{FGM}
Frenkel E., Givental A., Milanov T., Soliton equations, vertex operators, and
 simple singularities, \href{https://doi.org/10.1007/s11853-010-0035-6}{\textit{Funct. Anal. Other Math.}} \textbf{3} (2010),
 47--63, \href{https://arxiv.org/abs/0909.4032}{arXiv:0909.4032}.

\bibitem{HM}
Hollowood T., Miramontes J.L., Tau-functions and generalized integrable
 hierarchies, \href{https://doi.org/10.1007/BF02098021}{\textit{Comm. Math. Phys.}} \textbf{157} (1993), 99--117,
 \href{https://arxiv.org/abs/hep-th/9208058}{arXiv:hep-th/9208058}.

\bibitem{Kac78}
Kac V.G., Infinite-dimensional algebras, {D}edekind's {$\eta $}-function,
 classical {M}\"obius function and the very strange formula, \href{https://doi.org/10.1016/0001-8708(78)90033-6}{\textit{Adv.
 Math.}} \textbf{30} (1978), 85--136.

\bibitem{Kac94}
Kac V.G., Infinite-dimensional {L}ie algebras, 3rd ed., \href{https://doi.org/10.1017/CBO9780511626234}{Cambridge University
 Press}, Cambridge, 1990.

\bibitem{KW}
Kac V.G., Wakimoto M., Exceptional hierarchies of soliton equations, in Theta
 Functions--{B}owdoin 1987, {P}art 1 ({B}runswick, {ME}, 1987), \textit{Proc.
 Sympos. Pure Math.}, Vol.~49, Amer. Math. Soc., Providence, RI, 1989,
 191--237.

\bibitem{Kaup80}
Kaup D.J., On the inverse scattering problem for cubic eigenvalue problems of
 the class {$\psi_{xxx}+6Q\psi_{x}+6R\psi =\lambda \psi $}, \href{https://doi.org/10.1002/sapm1980623189}{\textit{Stud.
 Appl. Math.}} \textbf{62} (1980), 189--216.

\bibitem{Kos}
Kostant B., The principal three-dimensional subgroup and the {B}etti numbers of
 a complex simple {L}ie group, \href{https://doi.org/10.2307/2372999}{\textit{Amer.~J.~Math.}} \textbf{81} (1959),
 973--1032.

\bibitem{LRZ15}
Liu S.-Q., Ruan Y., Zhang Y., B{CFG} {D}rinfeld--{S}okolov hierarchies and
 {FJRW}-theory, \href{https://doi.org/10.1007/s00222-014-0559-3}{\textit{Invent. Math.}} \textbf{201} (2015), 711--772,
 \href{https://arxiv.org/abs/1312.7227}{arXiv:1312.7227}.

\bibitem{LWZ19}
Liu S.-Q., Wu C.-Z., Zhang Y., Virasoro {C}onstraints for {D}rinfeld--{S}okolov
 hierarchies and equations of {P}ainlev\'e type, \href{https://doi.org/10.1112/jlms.12603}{\textit{J.~London Math. Soc.}}
 \textbf{106} (2022), 1443--1500, \href{https://arxiv.org/abs/1908.06707}{arXiv:1908.06707}.

\bibitem{LWZZ20}
Liu S.-Q., Wu C.-Z., Zhang Y., Zhou X., Drinfeld--{S}okolov hierarchies and
 diagram automorphisms of affine {K}ac--{M}oody algebras, \href{https://doi.org/10.1007/s00220-019-03568-4}{\textit{Comm. Math.
 Phys.}} \textbf{375} (2020), 785--832, \href{https://arxiv.org/abs/1811.10137}{arXiv:1811.10137}.

\bibitem{M78}
Magri F., A simple model of the integrable {H}amiltonian equation,
 \href{https://doi.org/10.1063/1.523777}{\textit{J.~Math. Phys.}} \textbf{19} (1978), 1156--1162.

\bibitem{SK74}
Sawada K., Kotera T., A method for finding {$N$}-soliton solutions of the
 {K}.d.{V}. equation and {K}.d.{V}.-like equation, \href{https://doi.org/10.1143/PTP.51.1355}{\textit{Progr. Theoret.
 Phys.}} \textbf{51} (1974), 1355--1367.

\bibitem{Wu17}
Wu C.-Z., Tau functions and {V}irasoro symmetries for {D}rinfeld--{S}okolov
 hierarchies, \href{https://doi.org/10.1016/j.aim.2016.10.028}{\textit{Adv. Math.}} \textbf{306} (2017), 603--652,
 \href{https://arxiv.org/abs/1203.5750}{arXiv:1203.5750}.

\bibitem{Zhou}
Zhou J., On absolute {$N$}-point function associated with {G}elfand--{D}ickey
 polynomials, {U}npublished, 2015.

\end{thebibliography}

\pdfbookmark[1]{References}{ref}
\LastPageEnding

\end{document}